\documentclass[a4paper,fleqn]{article}

\usepackage{graphicx}
\usepackage[utf8]{inputenc}
\usepackage{a4}
\usepackage{amsfonts,amssymb,amsmath}
\usepackage[amsthm]{ntheorem}
\usepackage{stmaryrd}
\usepackage[square]{natbib}%\usepackage{ntheorem}

\newtheorem{theorem}{{\bf Theorem}}[section]
\newtheorem{proposition}[theorem]{{\bf Proposition}}

\theoremstyle{definition}
\newtheorem{definition}[theorem]{{\bf Definition}}%\theoremstyle{plain}
\newtheorem{claim}[theorem]{{\bf Claim}}
\theoremstyle{remark}

\newtheorem*{observation}{{\bf Observation}}
\newtheorem*{remark}{{\bf Remark}}
\newtheorem*{notation}{{\bf Notation}}
\newtheorem*{example}{{\bf Example}}
\def\llangle{\langle\!\langle} \def\rrangle{\rangle\!\rangle}
\def\eps{\varepsilon}

\usepackage{fancyhdr}
\begin{document}

\title{Local Commutators and Deformations in Conformal Chiral Quantum Field Theories\footnote{Supported in part by the German
    Research Foundation (Deutsche Forschungsgemeinschaft (DFG))
    through the Institutional Strategy of the University of
    G\"ottingen, and through the Research Training School 1493 ``Mathematical
    Structures in Modern Quantum Physics''}}

\author{Antonia M. Kukhtina, Karl-Henning Rehren \\[5mm] \normalsize 
Institute for Theoretical Physics, University of G\"ottingen \\ \normalsize 
Friedrich-Hund-Platz 1, D-37077 G\"ottingen \\[1mm] \normalsize and
\\[1mm] \normalsize Courant Centre ``Higher Order Structures in
Mathematics'', \\ \normalsize Bunsenstr. 3--5, D-37073 G\"ottingen} 
\maketitle

\abstract{We study the general form of M\"{o}bius covariant local
  commutation relations in conformal chiral quantum field theories and
  show that they are intrinsically determined up to structure
  constants, which are subject to an infinite system of constraints. The
  deformation theory of these commutators is controlled by a
  cohomology complex, whose cochain spaces consist of
  linear maps that are subject to a complicated symmetry property,
  a generalization of the anti-symmetry of the Lie algebra case.} 

\section{Introduction}

\normalsize
The theorem of L\"uscher and Mack \citep{LM1988}, which determined the
commutation relations of the stress--energy tensor, is an inspiring
example of how one can compute the commutators in conformal field
theory just on the basis of the most general properties of a
relativistic local quantum theory and conformal invariance. Using the
same argument one can fix the commutators of the stress--energy tensor
with an arbitrary primary field and one can
almost fix the commutators of the stress--energy tensor with a
quasiprimary field. We shall show that a similar strategy allows
to determine the commutation relations between arbitrary conformal
chiral fields (also known as ``$W$-algebras'') up to some structure
constants which we show to be subject to an infinite number of
constraints, reflecting anti--symmetry of commutators and the Jacobi
identity among smeared field operators. The solutions to these
constraints carry information about the specific model considered.

The anti--symmetry of commutators produces a symmetry rule for the
structure constants right away. However, the restrictions coming from
the Jacobi identity are not visible at once, because the different terms
there appear with different test functions and this does not allow us
to obtain relations only among the structure constants. To do this,
we study the effect of the commutator on the test function level and
observe that it gives rise to local intertwiners of the
$sl(2,\mathbb{R})$ action on the test function spaces. With the help
of transformation matrices of local intertwiners we achieve a
reduction of the field algebra, which means that we strip off the test
functions. This reduced structure has the form of a bilinear bracket
on a reduced field space. Apart from a mixed symmetry or
anti--symmetry of this bracket, its Jacobi identity involves certain   
coefficient matrices multiplying the three terms of the Jacobi
identity. These matrices are universal in the sense that they
reflect only the underlying representation theory of
$sl(2,\mathbb{R})$, but not the specific model.
They are thus constitutive elements of a new generalized
Lie--algebra-like bracket structure that can be used for the
classification of $W$-algebras.

These new identities constitute an infinite number of quadratic
constraints for the structure constants of $W$-algebras, not involving
the test functions any more. The solutions of these constraints
promote potential candidates for chiral conformal field theories. The
idea to consider constraints in such form was cherished from
\citep{bowcock91}, where a Jacobi identity among structure constants
from commutators of Fourier modes of quasiprimary fields was
considered.    
  
We then study the deformation theory of the commutators of the
reduced field algebra. The motivating example for us was
\citep{hollands08}, where deformations in the setting of the OPE
(operator product expansion) approach to quantum field theory on curved
space--time were studied. We consider deformations in a sense of
perturbative power series and work in a setting analogous to that in
\citep{gerstenhaber64}, which is the prototype of deformation theory
for algebraic structures. In all such theories the first step is to
relate the deformation problem to a certain cochain complex. In the
first examples of deformation theories of algebraic structures
\citep{gerstenhaber64}, \citep{nijenhuis67} the second step was to
show that the first cohomology groups are directly related to the
possibility to deform the algebraic structure considered. The more
modern point of view is that the deformation theory in consideration
is mastered by a differential graded Lie algebra (or in some cases a
homotopy Lie algebra or $L_{\infty}$-algebra) which can be obtained
from the cochain complex by constructing a bracket on this complex,
which is skew symmetric with respect to the grading by dimension
of the cochain spaces and satisfying a graded Jacobi identity
\citep{nijenhuis64}, \citep{manetti05}, \citep{borisov05}. 

The cochain complex, which we constructed, consists of
multilinear maps with a complicated permutation symmetry property ---
$Z^{\eps}$-symmetry (section \ref{sec:Zsym}). The origin of
this symmetry can be traced back to the symmetry rules in
the reduced algebra. We show that the first perturbations (also
infinitesimal perturbations) of the reduced brackets are classes from
the second cohomology group of our complex and we compute the
obstruction operators to their integration. We expect that an explicit
computation of the cohomology groups in the future will allow us to
relate the first of these groups to the problem of rigidity of the
bracket and the integrability of the first perturbations.

%The minimal models have a pure field content --- the only primary
%field there is the identity and the only quasiprimary is the
%stress--energy tensor. We want to explore theories with more fields,
%this is why we consider local extensions of minimal models. 
\section{Preliminaries}

The conformal group in a chiral theory is $\hbox{Diff}(S^1)$. It is
represented by a unitary representation $U$ on the Hilbert space of
the conformal field theory. A conformal chiral field $\Phi (z)$ on
$S^1$ transforms under a diffeomorphism $\gamma$ as a covariant tensor of
scaling dimension $d_\Phi$ if
\begin{equation}
 U(\gamma) \Phi(z) U^{-1}(\gamma) =
 \left(\frac{d\gamma}{dz}\right)^{d_{\Phi}} \Phi
 \left(\gamma(z)\right) \nonumber 
\end{equation}
holds. For local fields, the scaling dimension is an integer. Fields
which transform covariantly under the whole conformal group, are
called \emph{primary}. However, they do not exhaust the field content
of a theory. For example, in every conformal quantum field theory is
present the stress--energy tensor $T(x)$, which is responsible for
infinitesimal conformal transformations. $T(x)$ transforms covariantly
only under the M\"{o}bius subgroup $\hbox{SL}(2,\mathbb{R})$ of
$\hbox{Diff}(S^1)$, and such fields are called
\emph{quasiprimary}. Furthermore, in the OPE of a primary field with
$T(x)$ arise a series of other quasiprimary fields together with their
derivatives, and such fields are called \emph{secondary}. 

In all that follows, we identify $\mathbb{R}$ with
$S^1\setminus\{-1\}$ by the Cayley transform, and regard the fields as
distribution on $\mathbb{R}$. Since $A'(f)=-A(f')$, we don't 
consider the derivatives of quasiprimary fields as independent fields. 
Hence a basis of the field algebra is an infinite set of quasiprimary
fields. In a decent theory, e.g., such that $e^{-\beta L_0}$ is a
trace-class operator, the number of quasiprimary fields of a given
dimension is finite. We shall denote the basis of fields of
scaling dimension $a$ by $W_a$, and assume without loss of generality
that all $A\in W_a$ are hermitian fields. 

The commutators of the stress--energy tensor in a chiral theory are
intrinsically fixed: 
\begin{theorem}[L\"{u}scher--Mack \citep{LM1988}]
The stress--energy tensor in a chiral theory has the following
commutation relations: 
\begin{equation}
i[T(x),T(y)] =  T^{\prime}(y) \delta(x-y) -2T(y)\delta ^{\prime}(x-y)
+ \frac{c}{24} \delta ^{\prime \prime \prime}(x-y) \nonumber 
\end{equation}
\end{theorem}
With similar technique we find the commutator of $T(x)$ with some
arbitrary primary field $\varphi (x)$: 
\begin{equation}
 i[T(x),\varphi (y)]= \varphi^{\prime} (y) \delta (x-y) -
 d_{\varphi}\,\varphi(y) \delta ^{\prime}  (x-y) \nonumber 
\end{equation}
and with some arbitrary quasiprimary field $\phi (x)$:
\begin{equation}
 i[T(x),\phi (y)]= \phi^{\prime} (y) \delta (x-y) - d_{\phi}\,\phi(y)
 \delta ^{\prime}  (x-y) +\sum_{3 \leq k\leq h+1}\delta^{(k)}(x-y)
 \phi_k(y) \nonumber 
\end{equation}
where $\phi_k (x)$ are either quasiprimary fields or derivatives of
quasiprimary fields of lower dimensions. 

\section{The general form of local commutation relations in conformal
  chiral field theories} 

In this section we will show that the commutation relations in
conformal chiral field theories are intrinsically determined up to 
numerical factors (``structure constants'') by locality, conformal
invariance and Wightman positivity, and that the Lie algebra structure
imposes further constraints on the possible values of the structure constants. 

It will be enough to find just the commutators among the basis
quasiprimary fields. Our strategy to understand the general structure
of M\"{o}bius covariant commutators in chiral conformal field theories
is similar to that of the L\"{u}scher--Mack theorem: 

\begin{proposition}\label{prop:local}
Locality, scale invariance and Wightman positivity imply the following
general form of the commutator of two smeared quasiprimary field operators
$A(f)$ and $B(g)$: 
\begin{equation}
\label{eq:CR2}
-i \left[ A(f), B(g) \right] =  \sum_{c< a+b} \;\sum_{C\in W_c} F_{AB}^C\, C\left(\lambda_{ab}^c(f,g)\right) \,, 
\end{equation}
where $a$, $b$ are the scaling dimensions of $A$ and $B$, the sum runs
over a basis of quasiprimary fields of scaling dimension $c<a+b$,
$F_{AB}^C$ are numerical coefficients, and 
\begin{equation}
\label{eq:lamfg}
\lambda_{ab}^c(f,g) = \sum_{\substack{p,q \geq 0 \\ p+q=a+b-c-1}}
\lambda^c_{ab}(p,q) \,\partial^p f \cdot \partial^q g  
\end{equation}
are bilinear maps on the test functions that preserve supports, i.e.,
$\mathrm{supp}\,\lambda_{ab}^c(f,g)\subset \mathrm{supp}\,f
\cap\mathrm{supp}\,g$. These maps depend only on the dimensions of the
fields involved. 
\end{proposition}
\begin{proof}
We present here the main steps of the proof:
\begin{enumerate}
 \item Locality implies that the commutator $-i\left[A(x),B(y)\right]$ 
   has support on the line $x=y$. Then follows that $-i
   \left[ A(x), B(y) \right]=\sum_{l=0}^n
   \delta^{(l)}(x-y)O_l(y)$, where $O_l$ are linear combination of
   quasiprimary fields and their derivatives. This means that in the
   smeared version $-i \left[ A(f), B(g) \right]$ a quasiprimary field
   $C$ must appear with the test function of the form
   $\sum_{p,q \geq 0} d^{C}_{AB}(p,q) \partial^p f
   \cdot \partial^q g$. The coefficients $d^{C}_{AB}(p,q)$ satisfy a
   recursion in $p$ and $q$, coming from M\"{o}bius invariance, and
   the solution of this recursion is fixed, up to a factor, only by
   the scaling dimensions of the fields $A, \, B ,\, C$. The numerical 
   factor can be absorbed in the coefficients $F_{AB}^C$. 
 \item Scaling invariance implies that $p+q=n$ if $C(y)$ is a local field
   of scaling dimension $a+b-n-1$. 
 \item Wightman positivity implies that the scaling dimension of the
   fields on the theory must be non-negative (unitarity bound), hence $c \in
   [0,a+b-1]$. 
\end{enumerate} 
\end{proof}

\begin{observation}
The recursion for $\lambda_{ab}^c (p,q)$ coming from the
M\"{o}bius invariance (for fixed $a,b \geq 1$ and positive $c$) is solved by: 
\begin{equation}
\label{eq:lambda}
\lambda^c_{ab}(p,q)= (-1)^q \frac{(c+b-a)_p}{p!}\frac{(c+a-b)_q}{q!}\,\delta_{p+q,a+b-c-1}
\end{equation}
\end{observation}
where $(x)_n$ denotes the Pochhammer symbol:
\begin{equation}
 (x)_n:= \frac{\Gamma (x+n)}{ \Gamma (x)}.
\end{equation}
In particular, the maps $\lambda_{ab}^c (f,g)
= \sum_{p+q=a+b-c-1}\lambda_{ab}^c (p,q) \, \partial^p f
\cdot \partial^q g$ enjoy the graded symmetry property
\begin{equation}\label{eq:symm}
\lambda_{ab}^c (f,g) = (-1)^{a+b-c-1}\cdot\lambda_{ba}^c (g,f).
\end{equation}
Note that this (anti)symmetry respects the $\mathbb{Z}_2$ grading of
the source and range spaces, but the system of bilinear maps
$\lambda_{ab}^c$ themselves don't: there is no condition on
$c$ apart from $c<a+b$.

It is noteworthy to recognize that $\lambda_{ab}^c$ coincide with the
Rankin--Cohen brackets arising in the theory of modular forms. The
latter are bilinear differential maps $[f,g]_n:M_{2k}\times M_{2l}\to
M_{2k+2l+2n}$ on the spaces of modular forms of weights $2k$, $2l$
(\citep{rankin56,cohen75,cmz96}). In this context, of course, the test
functions have to be replaced by modular forms, and the emphasis is on
the discrete subgroup $\hbox{SL}(2,\mathbb{Z})$ of
$\hbox{SL}(2,\mathbb{R})$, under which modular forms are invariant.    
The precise relation is (with notations as in \citep{cmz96})
\begin{equation}
\label{eq:RC}
\lambda^c_{ab}(f,g) \equiv \left[f,g\right]_{n=a+b-c-1}^{(k=1-a,l=1-b)}\,. 
\end{equation}
We will give some more comments in App.\ \ref{sec:app}.

It becomes clear that the overall structure of the commutators in
conformal chiral field theories is to a great extent fixed -- we know
fields of which dimensions contribute to the commutator of any pair of
fields and with which test functions these fields are smeared. The
only unknown ingredients are the structure constants $F_{AB}^C$,
which are numbers. We shall now investigate further restrictions of
the structure constants due to the Lie algebra structure relations of
the commutator.  
\begin{observation}
The anti--symmetry of commutators together with the symmetry property
(\ref{eq:symm}) of $\lambda_{ab}^c$ implies the following symmetry
rule for the structure constants: 
\begin{equation}
\label{eq:Fgsym}
F^C_{AB} = (-1)^{a+b-c} \, F^C_{BA}
\end{equation}
\end{observation}
Taking adjoints, and recalling that the basis consists of hermitian
fields, one finds that $F^C_{AB}$ are real numbers.

Further restrictions for the structure constants $F_{AB}^C$ arise from
the Jacobi identity for commutators of smeared field operators, as we
will see in section \ref{sec:RJI}. We cannot derive these restrictions
directly, because the Jacobi identity in its original form would
produce constraints burdened with test functions. A reduction of the
field algebra, performed in section \ref{sec:redalg}, will allow us to
strip off the test functions from the Jacobi identity and to achieve a
reduced Jacobi identity involving only the structure constants $F_{AB}^C$. 
%To prepare the ground for that, in the next subsections we study the
%effect of the commutator on the test functions level.  

The $F_{AB}^C$ are also related to the amplitudes of $2$- and $3$-point
functions as we will elaborate in section \ref{sec:2p}. 

\subsection{$\lambda _{ab} ^c$ are intertwiners}

Quasiprimary fields of scaling dimension $a$ extend to a larger test
function space than just the Schwartz functions, namely to the space
$\pi_a$ of smooth functions on $\mathbb{R}$ for which
$x^{2-2a}f(x^{-1})$ extends smoothly to $x=0$. We regard this space as
a representation of $sl(2,\mathbb{R})$ with generators $p$, $d$, and
$k$ such that: 
\begin{equation}
 (pf)(x)=i\partial f(x), \quad \quad (df)(x) = i(x\partial + 1 - a)f(x),
 \quad \quad (kf)(x)=i(x^2 \partial +2(1 - a)x)f(x) 
\end{equation}
We must remark that $\pi_a$ is neither irreducible nor unitary. 
In particular, the inner product induced by the $2$-point function
annihilates the $(2a-1)$-dimensional subspace of polynomials of order
$2a-2$. 

The direct product $\pi_a \times \pi_b$ equals $\pi_a \otimes \pi_b$
as a space and carries the representation $(\pi_a \otimes \pi_b)
\circ \Delta$, where the $\Delta$ is the Lie algebra coproduct. 

Then the maps
\begin{equation}
\lambda^c_{ab}:\pi_a \times \pi_b \rightarrow \pi_c , \quad \quad
f \otimes g \mapsto \lambda_{ab}^c(f,g)=\sum_{p+q=a+b-c-1}
\lambda^c_{ab} (p,q) \, \partial^p f \cdot \partial^q g
\end{equation}
intertwine the corresponding $sl(2, \mathbb{R})$ actions on the spaces of test
functions. 
Their distinguishing feature among all such intertwiners is that they
preserve supports (see above), for which we call them 
\emph{local intertwiners}. The constructive argument in the proof of
Prop.\ \ref{prop:local} means that they are actually the unique local
intertwiners of the $sl(2,\mathbb{R})$ action. Therefore, our task
will be to understand the category of representations $\pi_a$ of
$sl(2,\mathbb{R})$ equipped with the local intertwiners.  

\subsection{Bases for the intertwiner spaces}
\label{sec:intsp}

One important observation is that the bound $c < a+b$ for $\lambda_{ab}^c$ 
guarantees that the intertwiner spaces $\pi_{a_1} \times \pi_{a_2}
\times ... \times \pi_{a_n} \rightarrow \pi_e$, where $e <
\sum_{i=1}^n a_i -n$, are finite--dimensional. In this
subsection we will construct bases for the intertwiner spaces and
describe the relevant matrices for a switch between bases. 

Our ``default'' choice of basis, adapted to the structures which
appear in our calculations (nested commutators), will be the following: 
\begin{definition}[Default basis for intertwiners 
$\pi_{a_1} \times \pi_{a_2} \times ... \times \pi_{a_n} \rightarrow \pi_e$]
%$\Big(T_{\underline{a}_n } \Big)^{\underline{m}_{n-1}}$] 
We define the operators:
%\footnotesize
\begin{eqnarray}
\label{eq:basT}
\Big(T_{\underline{a}_n} \Big)^{\underline{m}_{n-1}} =
\lambda^{e}_{a_1 \eps_1} \circ \left( 1_{a_1} \times 
\lambda^{\eps_1}_{a_2 \eps_2} \circ \left( 1_{a_1} 
\times 1_{a_2} \times \lambda^{\eps_2}_{a_3 \eps_3} 
\circ \left( ... \circ \left( 1_{a_1} \times ...
\times 1_{a_{n-2}} \times \lambda^{\eps_{n-2}}_{a_{n-1} a_n}
\right) ... \right)  \right) \right). 
\end{eqnarray}
%\normalsize
Here $\underline{x}_n$ stands for $n$-tuples $(x_1,...,x_n)$,
$\underline{a}_n$ is the $n$-tuple of scaling dimensions $a_i$ and the
indices $m_i \in \mathbb{N}_0$ are related to the scaling dimensions as: 
\begin{eqnarray}
 &&m_{n-1}:= a_{n-1} + a_n - \eps_{n-2} -1, \quad m_1 = a_1+ \eps_1 - e -1, \nonumber \\
 &&m_i:=a_i + \eps_{i} - \eps_{i-1} -1 \quad \hbox{for} \quad i=2...n-2
\end{eqnarray}
%The multi-index $\underline{e} _{n-1}$ of the basis elements cannot take arbitrary values. Its components 
Then the set of operators
$\big(T_{\underline{a}_n}\big)^{\underline{m}_{n-1}}$, such that  $m_1
+ ... + m_{n-1} = M(\underline{a}_n,e)\equiv \sum_{i=1} ^n a_i -
e-n+1$, constitute a basis for the intertwiner space $\pi _{a_1}
\times \pi _{a_2} \times ... \times \pi _{a_n} \rightarrow \pi _e$. 
%$m_1,...,m_{n-1}$ are varying such that their sum $m_1 + ... + m_{n-1} \in \, [0, \, \sum_i a_i - e_i-n+1]$.
\end{definition}

\begin{observation}
The $n-1$-tuple $\underline{m}_{n-1}$ determines the values of
the scaling dimensions $\eps_i$ and $e$ of the intermediate
and final representations:
\begin{equation}\label{eq:eps}
\eps_i =  \sum_{s=i+1} ^n a_s - \sum_{t=i+1}^{n-1} m_t
-n+i+1, \qquad e =  \sum_{s=1} ^n a_s - \sum_{t=1} ^{n-1} m_t -n+1
%\, , \qquad \quad e_i \leq a_n + a_n-1 + ... + a_i - n + i
\end{equation}
They are subject to restrictions, originating from the bound $c < a+b$
for $\lambda^c_{ab}$:  
\begin{equation}
 \eps_{n-2} \leq a_{n-1} + a_n -1, \quad \eps _1 \geq
 e-a_1 + 1, \quad \eps_i \leq \sum_{k=i+1}^n a_k - n+i +1
 \quad \hbox{for} \quad i=1...n-3 
\end{equation}
It should be noted that some of the dimensions $\eps_i$ may be 
negative. We shall ignore the unitarity bound (admitting only
nonnegative dimensions) at this point. It will be imposed later
(Sect.\ \ref{sec:axio}).
\end{observation}

\begin{remark}
The operators $\big(T_{\underline{a}_n}\big)^{\underline{m}_{n-1}}$
are multilinear maps on functions $(f_1,...,f_n)$ such that $f_i \in
\pi_{a_i}$. The images
$\big(T_{\underline{a}_n}\big)^{\underline{m}_{n-1}} (f_1,...,f_n)$  
are test functions belonging to the space $\pi_{e}$ ($e$ as in \eqref{eq:eps}). 
\end{remark}

Occasionally it will be necessary to consider nested brackets in
different order.
\begin{example}
An alternative basis for the intertwiner space $\pi_a \times \pi_b \times \pi_c \rightarrow \pi_e$ is:
\begin{equation}
 \Big( T_{S,abc} \Big)^{m_1m_2}:= \lambda_{\eps c}^e
 \circ(\lambda_{ab}^{\eps} \times 1_c), \quad
 m_1+m_2=M(a,b,c;e) = a+b+c-e-2
\end{equation}
\end{example}
In the general case, one may specify a ``bracket scheme'' $B$ and denote
the corresponding basis of intertwiners by 
$\big(T_{B, \underline{a}_n}\big)^{\underline{m}_{n-1}}$.

\subsubsection{Transformation matrices}\label{sec:XY}
From (\ref{eq:symm}) one immediately has
\begin{equation}\label{eq:flip}
 \Big( T_{abc} \Big)^{m_1m_2}(f,g,h) = (-1)^{m_1}\Big( T_{S,bca} \Big)^{m_1m_2}(g,h,f) = (-1)^{m_2} \Big( T_{acb} \Big)^{m_1m_2}(f,h,g).
\end{equation}

For the analysis of the Jacobi identity, however, we shall need
relations among $\big( T_{abc} \big)^{m_1m_2}(f,g,h)$ and 
$\big( T_{bca} \big)^{m_1m_2}(g,h,f)$ and $\big( T_{cab}
\big)^{m_1m_2}(h,f,g)$, not covered by (\ref{eq:flip}). In this
subsection we introduce the
transformation matrices for general permutations and re-bracketings.

%An important role in our computations will play the matrix $\big(Z _{ B_1B_2, \underline{a} _{n}, \sigma _{\underline{i} _n}}\big) ^{\underline{\widetilde{m}}_{n-1}} _{\underline{m}_{n-1}}$.
 
\begin{definition}[The matrix $\big(Z_{ B_1B_2, \underline{a}_{n},
    \sigma_{\underline{i}_{n}}}\big)^{\underline{\widetilde{m}}_{n-1}}_{\underline{m}_{n-1}}$] 
\label{def:Tb}
 Let us define the matrix $\big(Z_{ B_1B_2, \underline{a}_{n},
   \sigma_{\underline{i}_n}
 }\big)^{\underline{\widetilde{m}}_{n-1}}_{\underline{m}_{n-1}}$ which
 relates two bases $T_{B_1}$ and $T_{B_2}$ with permuted arguments:
\begin{equation} \label{eq:Tb}
 \Big(T_{B_1, \sigma_{\underline{i}_n}(
   \underline{a}_n)}\Big)^{\underline{\widetilde{m}}_{n-1}}\circ
 \tau_{\sigma_{\underline{i}_n}} = \Big(Z_{ B_1B_2, \underline{a}_{n}, \sigma_{\underline{i}_n}}\Big)
 ^{\underline{\widetilde{m}}_{n-1}}_{\underline{m}_{n-1}} \Big(T_{B_2,  \underline{a}_n}\Big)^{\underline{m}_{n-1}} 
\end{equation}
where $\sigma_{\underline{i}_n}$ is the permutation of labels $(x_1,...,x_n)
\mapsto (x_{i_1},...,x_{i_n})$ and $\tau_{\sigma_{\underline{i}_n}}:
\, (f_1,...,f_n) \mapsto (f_{i_1},...,f_{i_n})$ the corresponding
permutation on $\pi_{a_1}\times\dots\times\pi_{a_n}$.
In other words, permutations act on intertwiner spaces 
$\pi_{a_1}\times\dots\times\pi_{a_n}\to \pi_e$ by permutation of the factors, 
$\sigma(T):=T\circ \tau_\sigma$, and $\big(Z_{ B_1B_2, \underline{a}_{n}, 
\sigma_{\underline{i}_n}}\big)^{\underline{\widetilde{m}}_{n-1}}_{\underline{m}_{n-1}}$
are the  matrix elements of these linear maps between intertwiner
spaces in various bases of the latter. 
\end{definition}

Of particular interest for us will be the matrix
$\big(Y_{bca}\big)^{\widetilde{m}_1 \widetilde{m}_2}_{m_1 m_2}$
which describes the cyclic permutations of $\big(T_{abc}\big)^{m_1 m_2}
(f,g,h)$: 
\begin{equation}
\label{eq:Y}
  \Big(T_{bca}\Big)^{\widetilde{m}_1 \widetilde{m}_2} (g,h,f)=
\Big(Y_{bca} \Big)^{\widetilde{m}_1 \widetilde{m}_2}_{m_1 m_2}\Big(T_{abc}\Big)^{m_1 m_2} (f,g,h)
\end{equation}
By (\ref{eq:flip}), the transposition of the last two entries is
described by the diagonal matrix 
\begin{equation}
\label{eq:I}
I^{\widetilde{m}_1 \widetilde{m}_2}_{m_1 m_2} := \delta^{\widetilde{m}_1}_{m_1} \delta^{\widetilde{m}_2}_{m_2} (-1)^{m_2}.
\end{equation}
From the definition follows directly that $Y_{abc} \cdot Y_{cab} \cdot Y_{bca}
=1$ and $Y_{abc}\cdot I \cdot Y_{cba}\cdot I=1$, i.e., the matrices $Y$ and $I$
generate a representation of $S_3$. In particular, we have 
\begin{equation} \label{eq:transp}
T_{bac}(g,f,h)=I\,T_{bca}(g,h,f)=IY_{bca}\,T_{abc}(f,g,h).
\end{equation}
A calculation and explicit expression for the (quite complicated)
matrix elements $\big(Y_{abc}\big)_{\widetilde{m}_1 \widetilde{m}_2}^{m_1 m_2}$ 
can be found in the App.\ \ref{sec:app}. 

This matrix is closely related to the matrix that describes the
passage from the basis $\big(T_{abc}\big)^{m_1 m_2}$ to the basis
$\big(T_{S,abc}\big)^{m_1 m_2}$ without a permutation (``re-bracketing''): 
\begin{equation}
\Big(T_{abc}\Big)^{m_1m_2}(f,g,h)=\Big(X_{abc}\Big)^{m_1m_2}_{\widetilde{m}_1 \widetilde{m}_2} \Big(T_{S,abc}\Big)^{\widetilde{m}_1 \widetilde{m}_2}(f,g,h).
\end{equation}
Namely, by (\ref{eq:flip}), one has $T_{S,abc}(f,g,h)=
(-1)^{M}\,I\,T_{cab}(h,f,g) = (-1)^{M}\,I\,Y_{abc}^{-1}\,T_{abc}(f,g,h)$, where
$M=M(a,b,c;e) = a+b+c-e-2$ $(= m_1+m_2)$, hence 
\begin{equation} \label{X=YI}
X_{abc} = (-1)^{M(a,b,c;e)}\cdot Y_{abc}\, I.
\end{equation}

We claim that the matrix elements
$\big(Y_{abc}\big)_{\widetilde{m}_1 \widetilde{m}_2}^{m_1 m_2}$ are
the building blocks of every matrix element $(Z_{B_1B_2, \underline{a}_n,
  \sigma_{\underline{i}}})^{\underline{\widetilde{m}}_{n-1}}_{\underline{m}_{n-1}}$.  
Namely, one can achieve every bracket scheme from the default bracket
scheme (\ref{eq:basT}) by a sequence of applications of
(\ref{eq:symm}) (``flips''), at the price of a permutation of the
arguments. The flips will produce signs $(-1)^{m_i}$ where the label
$m_i$ refers to the flipped intertwiner. Now, the permutations can be
undone by a sequence of transpositions without changing the bracket
scheme. One sees from (\ref{eq:transp}) that in the default basis
(\ref{eq:basT}) the transposition $k\leftrightarrow k+1$ is described
by the matrix $\big(IY_{a_{k+1}\eps_{k+1} a_{k}}\big)^{m_k
  m_{k+1}}_{\widetilde m_k \widetilde m_{k+1}}\cdot \prod_{j\neq
  k,k+1}\delta^{m_j}_{\widetilde m_j}$.

\subsection{Reduction of the field algebra}
\label{sec:redalg} 

The field algebra, which we will denote with ${\cal V}$, decomposes as
a linear space into a direct sum of representations via commutators of
$sl(2,\mathbb{R})$, which is a subalgebra of ${\cal V}$: 
\begin{equation}
{\cal V} = \bigoplus_{a\in \mathbb{N}} {\cal V}_a
\end{equation}
% such that every subspace ${\cal V}_a$ is a span of  
Every subspace ${\cal V}_a$ is a span of (finitely many) quasiprimary
fields with the same integer scaling dimension $a>0$ and is isomorphic
to $V_a \otimes \pi_a$. As in subsection \ref{sec:intsp}, $\pi_a$ is a
test function space, which is a representation space for
$sl(2,\mathbb{R})$. 
$V_a$ is a finite--dimensional multiplicity space with basis $W_a$,
which accounts for the number of fields with scaling dimension $a$. 
The isomorphism above is realized by the map $\phi_a$ which acts as: 
\begin{equation}
\label{eq:iso}
\phi_a: \, A \otimes f \rightarrow A(f), \quad A \in V_a, \quad f \in \pi_a\,.
\end{equation}
We leave out the identity operator $I$ (of dimension $a=0$) from the
reduced space for several reasons: first, (\ref{eq:iso}) fails to
be an isomorphism in this case because $I(f) = (\int f(x)dx)\cdot 1$
depends only on the integral of $f$. Second, the unit operator is
central in the field algebra, so its commutator with other fields
contains no information. Third, the contribution of the unit operator
to the commutator of two fields is completely determined by the
$2$-point function, which we shall treat as an independent structure
element in Sect.\ \ref{sec:2p}. 

\begin{definition}[The reduced space $V$]
The direct sum of all multiplicity spaces $V= \bigoplus_{a\in
  \mathbb{N}} V_a$ will be called the reduced space $V$.  
\end{definition}

In the following we will show that the Lie algebra structure of 
${\cal V}$ is enciphered into multi--component structures on the reduced
space $V$. 
\begin{definition}[The reduced Lie bracket $\Gamma^{\ast} (\cdot , \cdot)_m$]
\label{def:rla}
On the reduced space $V = \bigoplus_a V_a$ the commutator
$[\cdot,\cdot]$ in ${\cal V}$ is represented by the
multi--component $\ast$-bracket $[\cdot,\cdot]^{\ast}_m$ or $\Gamma
^{\ast}(\cdot , \cdot)_m: V_a\times V_b\to V_{a+b-1-m}$, $m\geq0$: 
\begin{equation} \label{eq:rla} 
\Gamma ^{\ast}(A, B)_m := \sum_{C\in W_{a+b-1-m}} F_{AB}^C\, C\, .
\end{equation}
\end{definition}
Indeed, if we rewrite the Lie commutators (\ref{eq:CR2}) using
(\ref{eq:iso}) we find (suppressing the detailed form of the
contribution from the unit operator)  
\begin{eqnarray}
\label{eq:redbr}
-i[\phi_a (A \otimes f), \phi_b (B \otimes g)] & = & \sum_{c<a+b}
\phi_c \Big(\sum_{C\in W_c}F_{AB}^C \, 
C \otimes \lambda_{ab}^c (f,g)\Big) + (\hbox{unit operator})\nonumber \\ 
& = &\sum_{c<a+b} \phi_c \big(\Gamma^{\ast} (A,B)_{m=a+b-1-c} \otimes
  \lambda_{ab}^c (f,g) \big) + (\hbox{unit operator}).\qquad 
\end{eqnarray}
\begin{observation}
The anti--symmetry property of the commutator is encoded in the
graded symmetry property of the $\ast$-bracket: 
\begin{equation}
\label{eq:Ggsym}
\Gamma^{\ast}(X_1,X_2)_m=(-1)^{m+1} \Gamma ^{\ast}(X_1,X_2)_m \,.
\end{equation}
(\ref{eq:Ggsym}) actually reproduces the graded symmetry of the
structure constants $F_{AB}^C$ (\ref{eq:Fgsym}). 
\end{observation}

%Naturally comes the definition:
%\begin{definition}[The reduced Lie algebra $\{V, \Gamma ^{\ast}
%(\cdot , \cdot) _m\}$]
%The reduced space $V$ together with the reduced Lie bracket $\Gamma ^{\ast}
%(\cdot , \cdot) _m$ form an algebra, which we will call the reduced Lie algebra.
%\end{definition}
\begin{remark}
The reduction of the algebra may be interpreted as disentangling the $sl(2,
\mathbb{R})$ ``kinematic'' representation details from the structure
constants $F_{AB}^C$. The former are completely dictated by the
conformal symmetry, whereas the latter specify the model (together
with the dimensions $\mathrm{dim}\, V_a$). 
\end{remark}

In order to perform a complete reduction of the field algebra ${\cal
  V}$ we must also ``reduce'' the Jacobi identity and this will be
done in the next section.

\subsection{The reduced Jacobi identity and further constraints on $F_{AB}^C$}
\label{sec:RJI}

In this section we will examine what becomes of the Jacobi identity of
commutators under the ``space reduction''. In this way we will
complete the reduction of the field algebra and we will find further
restrictions on the coefficients $F_{AB}^C$. 

The Jacobi identity in its full form between three quasiprimary fields $A(f) \in {\cal V}_a$, $B(g) \in {\cal V}_b$ and $C(h) \in {\cal V}_c$ is:
\begin{equation}
  \Big[A(f), \big[B(g),C(h)\big]\Big] +  \Big[B(g), \big[C(h),A(f)\big]\Big]  +  \Big[C(h), \big[A(f),B(g)\big]\Big] = 0
\end{equation}
Now let us concentrate on the first term. As in (\ref{eq:redbr}), we
want to detach the test function contribution from the operator
part. Using the construction of intertwiners for multiple products of
representations (\ref{eq:basT}) and the relation (\ref{eq:iso}) we
write:
%Having in mind the relation, and  one gets exactly:
\begin{eqnarray} 
\Big[A(f), \big[B(g),C(h)\big]\Big] & \cong &
\sum_{m_1 m_2}  \Gamma^{\ast}\Big(A, \Gamma^{\ast} (B,C)\Big)_{m_1 m_2} \otimes \Big(T_{abc}\Big)^{m_1  m_2} (f,g,h) \\
\Gamma^{\ast} \Big(A, \Gamma^{\ast} (B,C)\Big)_{m_1 m_2} & = &
\sum_{\substack{E_1\in W_{e_1}= W_{a+b+c-m_1 - m_2 -2} \\  E_2 
\in W_{e_2} = W_{b+c-m_2-1}}} F^{E_1}_{A E_2} F^{E_2}_{BC} E_1  
\end{eqnarray}
Here and everywhere in the rest of this section the relation between
the $e$'s and the $m$'s are as in section \ref{sec:intsp}. 

The same considerations for the second and third terms yield similar
expressions but with $T_{bca}(g,h,f)$ and $T_{cab}(h,f,g)$ in the last
tensor factor. We then use (\ref{eq:Y}) to write them in the same form
$T_{abc}(f,g,h)$ as the first term. Then, by bilinearity of the tensor
product, the Jacobi identity reads
%\footnotesize
\begin{eqnarray} 
\sum_{m_1 m_2}  \Bigg\{ \Gamma ^{\ast} \Big(A, \Gamma ^{\ast} (B,C)\Big) _{m_1 m_2} +  \Gamma ^{\ast} \Big(B, \Gamma ^{\ast} (C,A)\Big) _{\widetilde{m} _1 \widetilde{m} _2} \Big(Y_{bca} \Big)^{\widetilde{m} _1 \widetilde{m} _2} _{m_1 m_2} 
 + \hspace{20mm}\\ \hspace{20mm}
 +  \Gamma ^{\ast} \Big(C, \Gamma ^{\ast} (A,B)\Big) _{\widehat{m} _1 \widehat{m} _2} \Big(Y_{cab} \Big)^{\widehat{m} _1 \widehat{m} _2} _{\widetilde{m} _1 \widetilde{m} _2} \Big(Y_{bca} \Big)^{\widetilde{m} _1 \widetilde{m} _2} _{m_1 m_2} \Bigg\} 
\otimes \Big(T _{abc}\Big)^{m _1  m _2} (f,g,h) = 0 \, . \nonumber 
\end{eqnarray}
%\normalsize
Having in mind that the basis components $\Big(T_{abc}\Big)^{m_1  m_2}
(f,g,h)$ for different values of $m_1$ and $m_2$ are linearly
independent functionals of the test functions, and the test functions
are arbitrary, we conclude for any fixed pair $(m_1,m_2)$: 
%\footnotesize
\begin{eqnarray} 
\label{eq:JI1}
 \Gamma ^{\ast} \Big(A, \Gamma ^{\ast} (B,C)\Big) _{m_1 m_2} +  \Gamma
 ^{\ast} \Big(B, \Gamma ^{\ast} (C,A)\Big) _{\widetilde{m} _1
   \widetilde{m} _2} \Big(Y_{bca} \Big)^{\widetilde{m} _1
   \widetilde{m} _2} _{m_1 m_2} + \hspace{40mm} \\ \hspace{40mm}
 +  \Gamma ^{\ast} \Big(C, \Gamma ^{\ast} (A,B)\Big) _{\widehat{m} _1
   \widehat{m} _2} \Big(Y_{cab} \Big)^{\widehat{m} _1 \widehat{m} _2}
 _{\widetilde{m} _1 \widetilde{m} _2} \Big(Y_{bca}
 \Big)^{\widetilde{m} _1 \widetilde{m} _2}_{m_1 m_2} =0 \, .\nonumber 
\end{eqnarray}
%\normalsize
Let us denote the left-hand-side of the reduced Jacobi identity
(\ref{eq:JI1}) with $\hbox{RJI}(A,B,C)_{m_1 m_2}$. Clearly, because 
$Y_{abc} \cdot Y_{cab} \cdot Y_{bca} =1$, one has the following symmetry rule: 
\begin{equation}
\label{eq:JIsym}
 \hbox{RJI}(A,B,C)_{m_1 m_2} =  \hbox{RJI}(B,C,A)_{\widetilde{m} _1 \widetilde{m} _2} \Big(Y_{bca} \Big)^{\widetilde{m}_1  \widetilde{m}_2}_{m_1 m_2}
\end{equation}
i.e., the vanishing of $\hbox{RJI}(A,B,C)_{m_1 m_2}$ is invariant
under cyclic permutations, as it should. 
If we use the explicit expressions for the nested $(\Gamma^{\ast})$'s
from above, the reduced Jacobi identity becomes for every quadruple of
quasiprimary fields $A,B,C$ and $E$ and for every pair $m_1,m_2$ such
that $m_1+m_2=a+b+c-e-2$:
%\footnotesize
\begin{eqnarray}
\label{eq:JI2}
\left[\sum_{E_2\in W_{e_2}}
F^{E}_{A E_2} F^{E_2}_{BC}\right]_{m_1 m_2} + \left[
\sum _{E_2\in W_{\tilde e_2}} F^{E} _{B E_2} F ^{E_2}_{CA}\right]_{\widetilde{m}_1 \widetilde{m}_2}  \Big(Y_{bca} \Big)^{\widetilde{m}_1 \widetilde{m}_2}_{m_1 m_2}  + \nonumber \\
+ \left[\sum _{E_2\in W_{\hat e_2}} F^{E}_{C E_2} F ^{E_2}_{AB}\right] _{\widehat{m}_1 \widehat{m}_2} \Big(Y_{cab} \Big)^{\widehat{m} _1 \widehat{m} _2} _{\widetilde{m} _1 \widetilde{m} _2} \Big(Y_{bca} \Big)^{\widetilde{m} _1 \widetilde{m} _2} _{m_1 m_2} = 0
\end{eqnarray}
%\normalsize
\begin{observation}
The reduced form of the Jacobi identity gives an infinite set of
constraints on the structure constants $F_{AB}^C$. Every solution of
this set of constraints promotes a candidate for the commutator
algebra of a local chiral conformal field theory.  

As noted in App.\ \ref{sec:app}, the matrix elements of $Y_{abc}$ can have
vanishing denominators, that have to be regularized (e.g., by giving
small imaginary parts to the dimensions). As it turns out,
in (\ref{eq:JI2}) these singularities will not be suppressed in
general by the vanishing of structure constants involving negative
scaling dimensions. To make sense of the singular Jacobi identities,
one has to multiply with the singular denominators and then remove the
regulators. The effect will be that only one or two of the three terms
of the Jacobi identity may survive, so that the general appearance of
the Jacobi identity may be quite different from the usual ``three-term''
form. Notice that anyway, due to the multi--component structure of the
bracket, each of the three terms is in general a sum over different
``intermediate'' representations. 
\end{observation}

\subsection{Relation between $F ^{A}_{BC}$ and $2$- and $3$-point amplitudes}
\label{sec:2p}

The $2$-point function of two hermitian fields $A(x)$ and $B(x)$ has
the form: 
\begin{equation}
 \langle A(x_1) B(x_2) \rangle = \llangle A B \rrangle
 \left( \frac{-i}{x_{12}  - i \eps} \right) ^{2a} \equiv
 \frac{\llangle A B \rrangle}{(ix_{12})_\eps^{2a}}\, .
\end{equation}
The map $A,B \mapsto \llangle A B \rrangle$ is a real
bilinear map on the reduced space which
\begin{itemize}
 \item is symmetric: $\llangle A B \rrangle = \llangle B A \rrangle$ ,
 \item respects the grading: $\llangle A B \rrangle = 0$
   if the scaling dimensions $a \neq b$ ,
 \item is positive definite: $\llangle A A \rrangle > 0$ unless $A=0$ .
\end{itemize}
The first property reflects locality of the QFT, the second is a
consequence of M\"obius invariance, and the last one is Wightman
positivity, i.e., the positive-definiteness of the Hilbert space inner
product. 

Similarly, the $3$-point function has the following form:
\begin{equation}
 \langle A(x) B(y) C(z) \rangle = \llangle A B  C\rrangle \frac{(-i)^{a+b+c}}{(x-y-i \eps) ^{a+b-c}(y-z-i \eps) ^{b+c-a}(x-z-i \eps) ^{a+c-b}}
\end{equation}
and by locality its amplitude must satisfy:
\begin{equation}
\llangle B A C\rrangle = (-1)^{a+b-c} \llangle A B  C\rrangle \,.
\end{equation}

We will show that the amplitudes of the $2$- and the $3$-point
functions are not independent on each other. For this purpose, let us
consider the $3$-point function $\langle [A,B] C \rangle$. We can find
it as $\langle [A,B] C \rangle = \langle AB C \rangle - \langle BA C
\rangle$. Using 
\begin{equation}
(-1)^n n! \left( \frac{1}{(x-i \eps)^{n+1}} - \frac{1}{(x+i \eps)^{n+1}} \right) = 2 \pi i \delta ^{(n)} (x)
\end{equation}
we obtain:
\begin{eqnarray}
\label{eq:3p1}
 \langle [A(x),B(y)] C(z) \rangle & = & \frac{(-i)^{a+b+c} \llangle A B  C\rrangle}{(x-z-i \eps) ^{2c}} \frac{2\pi i(-1)^m}{m!} \delta^{(m)} (x-y) + \nonumber \\
&& + \hbox{lower derivatives of }\delta\qquad (m\equiv a+b-c-1)\, .
\end{eqnarray}
On the other hand, taking into consideration $-i[A,B] = \sum
F_{AB}^C\, C$ and using the translation formula $C(\partial^p
f\cdot \partial^q g) \rightarrow (-1)^{p+q} \partial ^q_y( \partial _x
^p \delta (x-y) \cdot C(y))$ we end up with: 
\begin{eqnarray}
\label{eq:3p2}
 \langle [A(x),B(y)] C(z) \rangle & = & \sum_{C^{\prime}\in W_c} F
 _{AB} ^{C^{\prime}} (-1)^{m} \, \frac{(2c)_m}{m!}\,\delta^{(m)} (x-y)\llangle C^{\prime} C \rrangle \left(\frac{-i}{x-z-i \eps} \right) ^{2c} + \nonumber \\
 & & +  \hbox{lower derivatives of }\delta
\end{eqnarray}
Comparing (\ref{eq:3p1}) and (\ref{eq:3p2}) we obtain:
\begin{equation}
 \frac{(-i)^{a+b+c}}{(2c)_{a+b-c-1}} 2 \pi i \llangle ABC \rangle
 \rangle = \sum_{C^{\prime}\in W_c} F_{AB}^{C^{\prime}} \, \llangle C^{\prime} C \rrangle (-i)^{2c}\,.
\end{equation}
With the same considerations for $ \langle A [B,C] \rangle$ we obtain:
\begin{equation}
 \frac{(-i)^{a+b+c}}{(2a)_{b+c-a-1}} 2 \pi i \llangle ABC \rrangle = \sum
 _{A^{\prime}\in W_a} F_{BC}^{A^{\prime}}\, \llangle A^{\prime} A \rrangle (-i)^{2a}\,.
\end{equation}
The last two formulae allow us to find a new condition on the
structure constants $F_{AB}^C$ involving only $2$-point amplitudes: 
\begin{eqnarray} \label{eq:Finvariant}
(-1)^c (2c)_{a+b-c-1}\sum_{C^{\prime}\in W_c} F_{AB}^{C^{\prime}} \, \llangle C^{\prime} C
\rrangle = (-1)^a (2a)_{b+c-a-1}
\sum_{A^{\prime}\in W_a}
F_{BC}^{A^{\prime}} \, \llangle A^{\prime} A \rrangle \,. 
\end{eqnarray}
or
\begin{equation} \label{eq:Ginvariant}
 (-1)^c (2c)_{a+b-c-1} \llangle\Gamma^{\ast}(A,B)_{a+b-c-1},C\rrangle
 = (-1)^a (2a)_{b+c-a-1}\llangle A,\Gamma^{\ast}(B,C)_{b+c-a-1}\rrangle \,.
\end{equation}
There are two ways how to look at this condition: either one assumes a
given quadratic form $\llangle\cdot,\cdot\rrangle$, which
amounts to fixing bases of the finite-dimensional reduced field spaces
$V_a$: then (\ref{eq:Finvariant}) is indeed an additional constraint on
the structure constants $F_{AB}^C$. Or one regards the reduced algebra
(\ref{eq:rla}) subject to the structure relations (\ref{eq:Fgsym}) and
(\ref{eq:JI2}) as the primary structure: then (\ref{eq:Ginvariant}) is
an invariance condition on the quadratic form, in the same way as the
invariance condition $g([X,Y],Z)=g(X,[Y,Z])$ on a quadratic form on a
Lie algebra. This invariant quadratic form on the
reduced Lie algebra corresponds to the vacuum expectation
functional on the original commutator algebra. 

\subsection{Axiomatization of chiral conformal QFT}
\label{sec:axio}

The upshot of the previous analysis is a new axiomatization of chiral
conformal quantum field theory. It consists of the three data:

\begin{itemize}
\item a graded reduced space of fields $V = \bigoplus_{a \in\mathbb{N}} V_a$,
\item a generalized Lie bracket $\Gamma^{\ast} = \sum_{m\geq 0}
  \Gamma^{\ast}_m:V\times V\to V$, 
\item and a quadratic form $\llangle\cdot\,\cdot\rrangle:V\times V\to
  \mathbb{R}$. 
\end{itemize}
These data should enjoy the features outlined before: $V_a$ are real
linear spaces; the bracket is filtered: $\Gamma^{\ast}(V_a\times
V_b)\subset \bigoplus_{m\geq 0}V_{a+b-1-m}$, and satisfies the graded symmetry
(\ref{eq:Ggsym}) and generalized Jacobi identity (\ref{eq:JI1}); the
quadratic form is symmetric, positive definite, respects the grading,
and is invariant (\ref{eq:Ginvariant}) with respect to the bracket. 

Notice that the unitarity bound (absence of negative scaling
dimensions) has been imposed through the specification of the reduced
space $V$. Although the local intertwiner bases, and therefore also the
coefficient matrices $Y$ in the Jacobi identity do involve
``intermediate'' representations of negative dimensions ($a+b-1-m$ may
be $<0$), these do not contribute to the present axiomatization
because they multiply non-existent structure constants. Recall also
that the possibly singular instances of the Jacobi identity have to be
understood as explained in the end of Sect.\ \ref{sec:RJI}.

One may impose further physically motivated constraints, e.g., the
existence of a stress-energy tensor as a distinguished field $T\in
V_2$ whose structure constants $F_{TA}^A$ take canonical values; or
the generation of the entire reduced space by iterated brackets of a
finite set of fields, formulated as a surjectivity property of the
bracket. 

As a simple example, one may consider the constraints on the structure
constants for the commutator of two fields $A, B$ of dimension one. 
The only possibility in this case is $\mathrm{dim}\,C =1$. The
generalized Jacobi identity just reduces to the classical Jacobi
identity for the structure constants of some Lie algebra
$g$. Likewise, the invariance property of the quadratic form becomes
the classical $g$-invariance of the quadratic form 
$h(A,B)=\llangle AB\rrangle$ on $g$. The positivity condition on
the quadratic form implies that $g$ must be compact, and that $h$ is a
multiple of the Cartan-Killing metric. In other words: one obtains
precisely the Kac-Moody algebras as solutions to this part of the
constraints. The quantization of the level is expected to arise by 
the interplay between the positivity condition with the higher
generalized Jacobi identities. 

Other approaches \citep{zam86,bouw88,blum91} to the classification of
$W$-algebras 
have, of course, exploited essentially the same consistency relations
for a set of generating fields. Our focus here is, however, on the
entire structure including all ``composite'' fields, and the
possibility to formulate a deformation theory, to which we turn now.

\section{Cohomology of the reduced Lie algebra} \label{sec:coh}

In this section we will develop the cohomology of the reduced Lie
algebra as a prerequisite for the deformation theory in Sect.\
\ref{sec:def}. The description is intrinsic in the sense that it does
not refer to the commutator algebra of field operators is was derived
from. We generalize the lines of the cohomology theory of Lie algebras
\citep{chevalley48}, but the maps, which build the cochain spaces of
our cochain complex, will possess a more complicated symmetry
property, which we define now.  

\subsection{$Z _B ^{\eps}$-symmetry}
\label{sec:Zsym}

The reduced bracket (\ref{eq:redbr}) obeys the symmetry rule
(\ref{eq:Fgsym}). The reduced Jacobi identity (\ref{eq:JI1}) obeys the
symmetry rule (\ref{eq:JIsym}). A symmetry rule, generalizing the last
two rules for structures with more arguments, will be the following: 
\begin{definition}[$Z_{B}^{\eps}$-symmetry]
Let $V$ be the reduced space as in section \ref{sec:redalg} and let us
consider the maps $\omega _B ^{\ast n}
(\cdot,...,\cdot)_{m_1...m_{n-1}}:\, \underbrace{V \times ...\times V}_n
\rightarrow V$. Let $\underline{a}_n$
be an $n$-tuple of scaling dimensions $a_i$, let $X_i\in V_{a_i}$, and let
$\underline{m}_{n-1}$ will be the $n$-tuple $(m_1,...,m_{n-1})$. Let
$\omega _B ^{\ast n} (X_1,...,X_n)_{\underline{m}_{n-1}}$ be non-zero
only for $m_i \leq \sum _{s=i} ^n a_s - \sum _{t=i+1} ^{n-1} m_t -n
+i$. We will say that $\omega _B ^{\ast n}
(X_1,...,X_n)_{\underline{m}_{n-1}}$ 
%$\sum_{i=1} ^{n-1}m_i \leq \mathrm{min} \{a_i+a_j-1\}, \,i,j \in [0,n], \, i \neq j$
%$m_i \leq \sum _{s=i} ^n dim (X_s) - \sum _{t=i+1} ^{n-1} m_t -n +i$ and such that $\sum m_i \leq \sum \mathrm{dim}(X_k) - n + 1$, 
are $Z _B ^{\eps}$-symmetric if for every permutation in $S_n$
\begin{equation} \label{eq:Zesym}
 \omega _B ^{\ast n} (X_1,...,X_n) _{\underline{m}_{n-1}} = \omega _B
^{\ast n}(X_{i_1},...,X_{i_n})_{\underline{\widetilde{m}} _{n-1}} \Big(Z^{\eps} _{ BB,\underline{a} _n, \sigma _{\underline{i} _n}}\Big)
^{\underline{\widetilde{m}}_{n-1}} _{\underline{m}_{n-1}}
\end{equation}
where $\big(Z ^{\eps} _{ BB,\underline{a} _n, \sigma _{\underline{i} _n}}\big)^{\underline{\widetilde{m}}_{n-1}} _{\underline{m}_{n-1}}
:= \eps _{i_1...i_n}\big(Z _{ BB,\underline{a} _n, \sigma
  _{\underline{i} _n}}\big)^{\underline{\widetilde{m}}_{n-1}}
_{\underline{m}_{n-1}}$, with $\big(Z _{BB,\underline{a} _n, \sigma
  _{\underline{i} _n}}\big)^{\underline{\widetilde{m}}_{n-1}}
_{\underline{m}_{n-1}}$ the matrix representation of $S_n$ as in
definition \ref{def:Tb}, and  $\sigma_{\underline{i} _n}$ the
permutation $\{i_1,...,i_n\}$ of the indices $\{1,...,n\}$. 

\end{definition}
This definition is motivated by the following proposition:
\begin{proposition} The $Z _B ^{\eps}$-symmetry of $\omega_B^{\ast n} (X_1,...,X_n)_{\underline{m}_{n-1}}$ ensures that the function
\begin{equation}
\omega ^{n} (X_1(f_1),...,X_n(f_n)) := \sum _{\sum m_i < \sum a_k -n +1} \omega _B ^{\ast n} (X_1,...,X_n)_{\underline{m}_{n-1}} \otimes \Big(T _{B,\underline{a} _n}\Big)^{\underline{m}_{n-1}} (f_1,..., f_n)
\end{equation}
is completely anti--symmetric in the arguments $X_i(f_i)$.
%From now on we will attach to the $Z _B ^{\eps}$-symmetric maps an index $B$: $\omega _B %^{\ast n} (X_1,...,X_n)_{\underline{m}_{n-1}}$.
\end{proposition}
\begin{proof}
It follows directly from the definitions that:
\begin{eqnarray}
 \omega ^{n} (X_1(f_1),...,X_n(f_n)) & = & \sum  \omega _B ^{\ast n} (X_1,...,X_n)_{\underline{\widehat{m}}_{n-1}} \otimes \Big(T _{B,\underline{a} _n}\Big)^{\underline{\widehat{m}}_{n-1}} (f_1,..., f_n) \nonumber \\
& = & \sum \omega _B ^{\ast n} (X_1,...,X_n)_{\underline{\widehat{m}}_{n-1}}  \otimes \Big(Z ^{\eps} _{ B,\underline{a} _n, 1}\Big)^{\underline{\widehat{m}}_{n-1}} _{\underline{m}_{n-1}}\Big(T _{\underline{a} _n}\Big)^{\underline{m}_{n-1}} (f_1,..., f_n) \nonumber \\
%& = & \sum \omega _B ^{\ast n} (X_1,...,X_n)_{\underline{\widehat{m}}_{n-1}} \Big(Z ^{\eps} _{ B,\underline{a} _n, 1}\Big)^{\underline{\widehat{m}}_{n-1}} _{\underline{m}_{n-1}}  \otimes \Big(T _{\underline{a} _n}\Big)^{\underline{m}_{n-1}} (f_1,..., f_n) \nonumber \\
& = & \sum \omega _B
^{\ast n}(X_{i_1},...,X_{i_n})_{\underline{\widetilde{m}} _{n-1}}  \otimes \Big(Z^{\eps} _{ B,\underline{a} _n, \sigma _{\underline{i} _n}}\Big)
^{\underline{\widetilde{m}}_{n-1}} _{\underline{m}_{n-1}} \Big(T _{\underline{a} _n}\Big)^{\underline{m}_{n-1}} (f_1,..., f_n) \nonumber \\
%& = & \sum  \omega _B
%^{\ast n}(X_{i_1},...,X_{i_n})_{\underline{\widetilde{m}} _{n-1}} \Big(Z^{\eps} _{ B,\underline{a} %_n, \sigma _i}\Big)
%^{\underline{\widetilde{m}}_{n-1}} _{\underline{m}_{n-1}}  \otimes \Big(T _{\underline{a} _n}\Big)^{\underline{m}_{n-1}} (f_1,..., f_n)
& = & \sum \omega _B
^{\ast n}(X_{i_1},...,X_{i_n})_{\underline{\widetilde{m}} _{n-1}} \eps _{i_1...i_n}  \otimes \Big(T _{B,\sigma_{\underline{i} _n}(\underline{a} _n)}\Big)^{\underline{\widetilde{m}}_{n-1}} (f_{i_1},..., f_{i_n}) \nonumber \\ 
& = & \eps _{i_1...i_n} \omega ^{\ast n} (X_{i_1}(f_{i_1}),...,X_{i_n}(f_{i_n}))
\end{eqnarray}
which proves the proposition.
\end{proof}
\begin{notation}
We will be interested in those $Z_B ^{\eps}$-symmetric
maps for which $B$ is the default basis $T$ as in section
\ref{sec:intsp}. We will call such maps $Z^{\eps}$-symmetric. 
\end{notation}
\begin{example}
The two natural examples for $Z^{\eps}$-symmetric maps are the reduced bracket and the reduced Jacobi identity. 
\end{example}

\subsection{Reduced Lie algebra cohomology}
\label{sec:coho}

In this section we will introduce the reduced Lie algebra cohomology complex:

\begin{definition}[reduced Lie algebra cohomology]
We define the reduced Lie algebra cohomology as:
\begin{itemize}
 \item{ \emph{\textbf{Cochain complex:}}}
\end{itemize}
\begin{enumerate}
 \item{ \emph{\textbf{Cochain spaces $C^n(V)$ of dimension $n$:}}

The $n$-cochains in the cochain complex are the tensor--valued (i.e.,
multi-component) $Z^{\eps}$-symmetric maps \\ 
$\omega  ^{\ast n} (\cdot,...,\cdot)_{\underline{m}_{n-1}}$. The
spaces $C^n(V)$ of all $Z ^{\eps}$-symmetric $\omega ^{\ast
  n}$'s for a fixed $n$ compose the cochain sequence
$C:=(C^n(V))_{n \in \mathbb{N}_0}$.} 
 \item{\emph{\textbf{Coboundary operators $b^n$:}}

We define the coboundary operator 
%$b^n:=\{b^n _{m_i}\}_{m_i \in \mathbb{N}}: \, C^n(V) \rightarrow C^{n+1}(V)$ 
$b^n: \, C^n(V) \rightarrow C^{n+1}(V)$ through the following component--wise action, provided that $m_i \leq \sum _{s=i} ^n a_s - \sum _{t=i+1} ^{n-1} m_t -n +i$:
\footnotesize
\begin{eqnarray}
\label{eq:dif1}
[b^{n} \omega  ^{\ast n}](\underline{X}_{n+1})_{\underline{m}_n} & := & \frac{(-1)^n}{n!} \sum _{\underline{i}_{n+1}}
\Bigg[\Gamma ^{\ast} \Big( X_{i_1},\omega ^{\ast n}(X_{i_2},...,X_{i_{n+1}})\Big)\Bigg] _{\underline{\widetilde{m}} _{n}} \Big(Z^{\varepsilon} _{ \underline{a} _{n+1},  \sigma _{\underline{i} _{n+1}}}\Big)
^{\underline{\widetilde{m}}_{n}} _{\underline{m}_{n}} + \nonumber \\
&& + \frac{1}{2(n-1)!}  \sum _{\underline{j}_{n+1}}\Bigg[
\omega ^{\ast n}\Big(X_{j_1},...,X_{j_{n-1}},\Gamma ^{\ast }(X_{j_n},X_{j_{n+1}})\Big)\Bigg]_{\underline{\widetilde{m}} _{n}}  \Big(Z^{\varepsilon} _{\underline{a} _{n+1},  \sigma _{\underline{j} _{n+1}}}\Big)
^{\underline{\widetilde{m}}_{n}} _{\underline{m}_{n}} \nonumber \\
\end{eqnarray}
\normalsize
or equivalently (because $\omega ^{\ast n}$ is $Z^\eps$-symmetric):
\footnotesize
\begin{eqnarray}
\label{eq:dif2}
[b^n \omega ^{\ast n}](\underline{X}_{n+1})_{\underline{m}_n} & := & (-1)^n \sum _{i=1} ^{n+1} \Bigg[\Gamma ^{\ast} \Big( X_{i},\omega  ^{\ast n}(X_1,...,\hat{X} _i,...,X_{n+1})\Big) \Bigg] _{\underline{\widetilde{m}} _{n}} \Big(Z^{\varepsilon} _{ \underline{a} _{n+1}, \sigma _{\widehat{i}} }\Big)
^{\underline{\widetilde{m}}_{n}} _{\underline{m}_{n}} + \nonumber \\
&& + \sum _{k>j=1} ^{n} \Bigg[
\omega  ^{\ast n}\Big(X_1,...,\hat{X} _j,...,\hat{X} _k,...,X _{n+1},\Gamma ^{\ast }(X_{j},X_{k})\Big)\Bigg]_{\underline{\widetilde{m}} _{n}} \Big(Z^{\varepsilon} _{ \underline{a} _{n+1}, \sigma _{\widehat{j}\widehat{k}}}\Big)
^{\underline{\widetilde{m}}_{n}} _{\underline{m}_{n}} \nonumber \\
\end{eqnarray}
\normalsize
where $\sigma _{\widehat{i}}\in S_{n+1}$ is the permutation
$\{i,1,...,\widehat{i},...,n+1\}$ and $\sigma
_{\widehat{j}\widehat{k}}$ is the permutation
$\{1,...,\widehat{j},...,\widehat{k},...,n+1,j,k\}$ . 

Here and below we write the sum over permutations as $\sum _{\underline{i}_{n+1}}:= \sum _{\substack{i_1 \neq ... \neq i_{n+1} \\ i_k \in [1,...n+1]}}$.
%such that $m_1$ indices the component of $b^n$ and $m_2,...,m_n$ are the indices of $\omega ^{\ast n}$ before the transformation.

For those $n$-tuples $(m_1,...,m_n)$, for which the condition $m_i \leq \sum _{s=i} ^n a_s - \sum _{t=i+1} ^{n-1} m_t -n +i$ does not hold, $[b^{n} \omega  ^{\ast n}](X_1,...,X_{n+1})_{\underline{m}_n}$ will be set to $0$.

We will show below that $b^{n+1} \circ b^n = 0$.
%, which will mean $b^{n+1} _{m_1}\circ b^n _{m_2}= 0, \, \forall m_1,m_2$.
}
\end{enumerate}
\begin{itemize}
\item{\emph{\textbf{Cohomology group:}}

We define:
\begin{eqnarray}
Z^n (V) & := & \hbox{Ker} (b^n) = \Big\{ \omega  ^{\ast n} \in C^n (V) \mid \quad [b^{n} \omega  ^{\ast n}](X_1,...,X_{n+1})_{\underline{m}_n} = 0, \, \forall \underline{m}_n \in \mathbb{N}_0^{\otimes n} \Big\} \nonumber \\
B^n (V) & := & \hbox{Im} (b^n) = \Big\{ \omega ^{\ast n} \in C^n (V) \mid \quad \omega ^{\ast n} = b^{n-1} \omega ^{\ast n-1}, \, \omega ^{\ast n-1} \in C^{n-1}(V)\Big\}
\end{eqnarray}
$b^{n+1} \circ b^n = 0$ implies $B^n (V) \subseteqq Z^n
(V)$.
Then we define the $n^{\rm th}$ reduced Lie algebra cohomology group as the quotient:
\begin{equation}
RLH^n (V) = Z^n (V) / B^n (V)\,.
\end{equation} 
In writing $Z^n(V)$, $B^n(V)$ and $RLH^n(V)$, it is understood that
$V$ is equipped with a bracket $\Gamma^\ast$, on which these spaces
clearly depend.}
\end{itemize} 

\end{definition}
We have to prove that $b^n$ are differentials, i.e., 
\begin{proposition}
\label{prop:b}
 $b^{n+1} \circ b^n = 0$ applied to any map $\omega ^{\ast n}$ from the cochain complex.
\end{proposition}

\begin{proof}
The proof proceeds in perfect analogy with the cohomology of Lie
algebras, where the various terms obtained by evaluating $b^{n+1}
\circ b^n$ can be seen to cancel each other by virtue of the
antisymmetry of the Lie bracket, the Jacobi identity, and the
antisymmetry of the co-chains. However, there arises one salient
complication:

Due to the $Z^\eps$-symmetrization, there arise in the second line of
(\ref{eq:dif1}) terms of the structure
\begin{eqnarray} 
\omega(X,...,X,\Gamma(X,X)_{m_n},X,...,X)_{\underline m_{n-1}}.
\end{eqnarray}
Their multi-indices $\underline m_n = (\underline m_{n-1},m_n)$
correspond to non-default bracket scheme $B_\kappa$ (where
$\kappa=1,...,n$ is the position of the insertion) with intertwiner
basis of the structure 
\begin{eqnarray} 
T_{\underline m_{n-1}}\circ(1\times...\times 1\times
\lambda\times1\times...\times 1).
\end{eqnarray}
The necessary change of basis can be included into the matrix $Z^\eps$
by virtue of 
\begin{eqnarray} 
Z^\eps_{B_\kappa,\sigma(\underline a_n),1} Z^\eps_{\underline a_n,\sigma} 
= Z^\eps_{B_\kappa,\underline a_n,\sigma}.
\end{eqnarray}
One can then re-write (\ref{eq:dif1}) as  
\begin{eqnarray}
\label{eq:diff1}
[b^{n} \omega  ^{\ast n}](\underline{X}_{n+1})_{\underline{m}_n} = \frac{(-1)^n}{n!}\sum _{\underline{i} _{n+1}}
\Bigg[\Gamma ^{\ast} \Big( X_{i_1},\omega ^{\ast n}(X_{i_2},...,X_{i_{n+1}})\Big)\Bigg] _{\underline{\widetilde{m}} _{n}} \Big(Z^{\varepsilon} _{ \underline{a} _{n+1}, \sigma _{\underline{i} _{n+1}}}\Big)
^{\underline{\widetilde{m}}_{n}} _{\underline{m}_{n}} + \nonumber \\
 + \frac{1}{2(n-1)!}  \sum _{\kappa =1}^n\sum _{\underline{j} _{n+1}}\Bigg[
\omega ^{\ast n}\Big(X_{j_1},...,X_{j_{n-1}},\Gamma ^{\ast }(X_{j_n},X_{j_{n+1}})\Big)\Bigg]_{B_{\kappa},\underline{\widetilde{m}} _{n}}  \Big(Z^{\varepsilon} _{B_{\kappa}, \underline{a} _{n+1}, \sigma _{\underline{j} _{n+1}}}\Big)
^{\underline{\widetilde{m}}_{n}} _{\underline{m}_{n}} .\nonumber \\
\end{eqnarray}
where the term $\Big[
\omega ^{\ast n}\Big(X_{j_1},...,X_{j_{n-1}},\Gamma ^{\ast
}(X_{j_n},X_{j_{n+1}})\Big)\Big]_{B_{\kappa},\underline{\widetilde{m}}
  _{n}}$ collects all contributions, where $\Gamma ^{\ast
}(X_{j_n},X_{j_{n+1}})$ was inserted in the $\kappa^{\rm th}$
position. 

Then, when composing $b^{n+1}\circ b^n$, one will encounter also
bracket schemes $B_{\kappa_1,\kappa_2}$ and $\widetilde B_{\kappa}$,
$\widehat B_{\kappa}$ corresponding to intertwiner bases of the structure
\begin{eqnarray} 
T_{\underline m_{n-1}}&\circ&(1\times...\times 1\times
\lambda\times 1\times...\times 1\times\lambda\times 1\times...\times
1), \nonumber \\
T_{\underline m_{n-1}}&\circ&(1\times...\times 1\times
(\lambda\circ (1\times \lambda))\times...\times 1), 
\\
T_{\underline m_{n-1}}&\circ&(1\times...\times 1\times
(\lambda\circ (\lambda\times 1))\times...\times 1), \nonumber
\end{eqnarray}
respectively, where $\kappa$ and $\kappa_i$ stand for the positions of
the insertions. Even if essentially straightforward, the precise details
are quite cumbersome and will not be presented here, see
\citep{Kukhtina11}. 

\newpage

The result can then be written in the form 

\footnotesize
\begin{eqnarray}
\label{eq:bigpr}
&& (b^{n+1} \circ b^n \omega ^{\ast n})(X_1,...,X_{n+2})_{\underline{m}_{n+1}}  = \nonumber \\
&(\textbf{\hbox{I}})&-\frac{(n+1)}{(n+1)!} \, \sum _{\underline{i}_{n+2}} \bigg[ \Gamma ^{\ast} \bigg( X_{i_1}, \Gamma ^{\ast} \Big( X_{i_2},
\omega ^{\ast n}(X_{i_3},...,X_{i_{n+2}})\Big) \bigg)\bigg] _{\underline{\widetilde{m}} _{n+1}} \Big(Z^{\varepsilon} _{\underline{a} _{n+2}, \sigma _{\underline{i}_{n+2}}}\Big)
^{\underline{\widetilde{m}}_{n+1}} _{\underline{m}_{n+1}} + \nonumber \\
&(\textbf{\hbox{II}})& +(-1)^{n+1}\frac{(n+1)n}{2(n+1)!} \,\sum _{\kappa = 1} ^{n} \sum _{\underline{i}_{n+2}} \Gamma ^{\ast} \bigg( X_{i_1} ,\bigg[ \omega  ^{\ast n} \Big( X_{i_2},...,X_{i_{n}},
\Gamma ^{\ast}(X_{i_{n+1}},X_{i_{n+2}})\Big) \bigg] _{B_{\kappa}} \bigg)_{\underline{\widetilde{m}} _{n+1}} \Big(Z^{\varepsilon} _{B_{\kappa +1},\underline{a} _{n+2}, \sigma _{\underline{i}_{n+2}}}\Big)
^{\underline{\widetilde{m}}_{n+1}} _{\underline{m}_{n+1}} + \nonumber \\
&(\textbf{\hbox{III}})& +(-1)^{n}\frac{n}{2(n!)} \,\sum _{\kappa = 1} ^{n} \sum _{\underline{i}_{n+2}} \Gamma ^{\ast} \bigg( X_{i_1} , \bigg[\omega ^{\ast n} \Big( X_{i_2},...,X_{i_{n}},
\Gamma ^{\ast}(X_{i_{n+1}},X_{i_{n+2}})\Big)\bigg] _{B_{\kappa}} \bigg)_{\underline{\widetilde{m}} _{n+1}} \Big(Z^{\varepsilon} _{B_{\kappa+1},\underline{a} _{n+2}, \sigma _{\underline{i}_{n+2}}}\Big)
^{\underline{\widetilde{m}}_{n+1}} _{\underline{m}_{n+1}} + \nonumber \\
&(\textbf{\hbox{IV}})& +\frac{1}{2(n!)} \, \sum _{\underline{i}_{n+2}} \Gamma ^{\ast} \bigg( \Gamma ^{\ast} (X_{i_1},X_{i_2}),\omega ^{\ast n} ( X_{i_3},...
,X_{i_{n+2}}) \bigg) _{B_1, \underline{\widetilde{m}} _{n+1}}  \Big(Z^{\varepsilon} _{B_1, \underline{a} _{n+2}, \sigma _{\underline{i}_{n+2}}}\Big)
^{\underline{\widetilde{m}}_{n+1}} _{\underline{m}_{n+1}} + \nonumber \\
&(\textbf{\hbox{V}})& +\frac{n(n-1)}{4(n!)} \, \sum _{\kappa_1 \neq \kappa_2 = 1} ^{n}\sum _{\underline{i}_{n+2}} \Bigg[ \omega^{\ast n} \bigg( X_{i_1} ,...,\Gamma ^{\ast} ( X_{i_{n-1}}, X_{i_{n}})
,\Gamma ^{\ast}(X_{i_{n+1}},X_{i_{n+2}}) \bigg)\Bigg] _{B_{\kappa _1 \kappa _2},\underline{\widetilde{m}} _{n+1}} \Big( Z^{\varepsilon} _{B_{\kappa _1 \kappa _2}, \underline{a} _{n+2}, \sigma _{\underline{i}_{n+2}}}\Big)
^{\underline{\widetilde{m}}_{n+1}} _{\underline{m}_{n+1}} + \nonumber \\
&(\textbf{\hbox{VI}})& +\frac{n}{2(n!)} \, \sum _{\kappa = 1} ^{n} \sum _{\underline{i}_{n+2}} \Bigg[ \omega ^{\ast n} \bigg( X_{i_1},...,X_{i_{n-1}} ,\Gamma ^{\ast} \Big( X_{i_{n}},
\Gamma ^{\ast}(X_{i_{n+1}},X_{i_{n+2}})\Big) \bigg)\Bigg] _{\widetilde{B}_{\kappa},\underline{\widetilde{m}} _{n+1}} \Big(Z^{\varepsilon} _{ \widetilde{B}_{\kappa},\underline{a} _{n+2},\sigma _{\underline{i}_{n+2}}}\Big)
^{\underline{\widetilde{m}}_{n+1}} _{\underline{m}_{n+1}} \nonumber \\
\end{eqnarray}
\normalsize
The terms with the bracket scheme $\widehat B_\kappa$ are equal to
those with the bracket scheme $\widetilde B_\kappa$ by virtue of the
symmetry of $\Gamma^\ast$, and are included in the term {\bf (VI)}. 

Due to symmetry properties of $\omega^{\ast n}, \, \Gamma^{\ast}$ and
$Z_B^{\eps}$-matrices one then realizes that: 

\begin{itemize}
 \item{$\textbf{(II)} + \textbf{(III)} = 0$ term by term.}  
 \item{\textbf{(I)} + \textbf{(IV)} combine to a Jacobi identity
     between $X_{i_1}$, $X_{i_2}$ and $\omega ^{\ast n}(X_{i_3},...,X_{i_{n+2}}) 
     _{\underline{m}_{n-1}}$. This cancels them. 
}
\item{the terms in \textbf{(V)} can be rewritten as a sum of pairs of
    terms (with $\kappa_1$ and $\kappa_2$ exchanged) which cancel each
    other.  
}
\item{for each $\kappa$, \textbf{(VI)} can be grouped in triples
    (by cyclic permutations of $i_n,i_{n+1},i_{n+2}$) that contain a
    Jacobi identity which cancels them. 
}
\end{itemize}
Then the sum of all terms is $0$ and this proves the proposition.

\end{proof}

%\hspace{5mm}

\section{Deformations of the reduced Lie algebra} \label{sec:def}
 
We now consider formal deformations of the bracket of a reduced Lie
algebra, which are defined as a perturbative series, such that the
reduced Jacobi identity is respected. Our approach generalizes the
cohomological analysis of deformations of associative algebras 
\citep{gerstenhaber64}. 

\begin{definition}[Formal deformations of the reduced bracket]
A formal deformation of the bracket $\Gamma  ^{\ast} : V \otimes V
\rightarrow V$ is defined as a one-parameter family of brackets $\Gamma ^{\ast} (A,B, \lambda) _m$ with
$\lambda \in \mathbb{R}$ and $\Gamma ^{\ast} (A,B, 0) _m \cong \Gamma
^{\ast} (A,B) _m$. The deformed bracket is defined as a
formal power series:
\begin{equation} \label{eq:pert}
 \Gamma ^{\ast} (A,B, \lambda) _m := \sum _{i=0} ^{\infty} \Gamma _i ^{\ast} (A,B) _m \, \lambda ^i
\end{equation}
and the $i^{\rm th}$ order perturbations of the bracket is:
\begin{equation}
\Gamma _i ^{\ast} (A,B) _m :=\frac{1}{i!}\frac{d^i}{d \lambda ^i}  \Gamma ^{\ast}(\lambda) (A,B) _m
\end{equation}
Here $\Gamma _0 ^{\ast} (A,B) _m \equiv \Gamma ^{\ast} (A,B) _m$.
\end{definition}

We are interested only in those deformations which are consistent with the generalized Jacobi identity
(\ref{eq:JI1}). This leads to a number of constraints which single out the admissible
perturbations. 

The first order perturbations $\Gamma _1 ^{\ast} (A,B) _m$ must obey:
\footnotesize
\begin{eqnarray}
\label{eq:JIdef0}
 \Gamma _0 ^{\ast}  \Big(A, \Gamma _1 ^{\ast} (B,C)\Big)_{m_1 m_2} +  \Gamma _0 ^{\ast} \Big(B, \Gamma _1 ^{\ast} (C,A)\Big) _{\widetilde{m} _1 \widetilde{m} _2} \Big(Y_{bca} \Big)^{\widetilde{m} _1 \widetilde{m} _2} _{m_1 m_2}
+ \Gamma _0 ^{\ast} \Big(C, \Gamma _1 ^{\ast} (A,B)\Big) _{\widehat{m} _1 \widehat{m} _2} \Big(Y_{cab} \Big)^{\widehat{m} _1 \widehat{m} _2} _{\widetilde{m} _1 \widetilde{m} _2} \Big(Y_{bca} \Big)^{\widetilde{m} _1 \widetilde{m} _2} _{m_1 m_2}  + \nonumber \\
+ \Gamma _1  ^{\ast}\Big(A, \Gamma _0 ^{\ast} (B,C)\Big) _{m_1 m_2} +  \Gamma _1 ^{\ast}\Big(B, \Gamma _0 ^{\ast} (C,A)\Big) _{\widetilde{m} _1 \widetilde{m} _2} \Big(Y_{bca} \Big)^{\widetilde{m} _1 \widetilde{m} _2} _{m_1 m_2}
+ \Gamma _1 ^{\ast} \Big(C, \Gamma _0 ^{\ast} (A,B)\Big) _{\widehat{m} _1 \widehat{m} _2} \Big(Y_{cab} \Big)^{\widehat{m} _1 \widehat{m} _2} _{\widetilde{m} _1 \widetilde{m} _2} \Big(Y_{bca} \Big)^{\widetilde{m} _1 \widetilde{m} _2} _{m_1 m_2}  = 0 \nonumber \\
\end{eqnarray}
\normalsize
The higher order perturbations must satisfy the following condition:
\begin{eqnarray}
\label{eq:JIdefj}
\sum _{k=0} ^{n} \, \Bigg\{ \Gamma _k ^{\ast} \Big(A, \Gamma _{n-k} ^{\ast} (B,C)\Big) _{m_1 m_2} +  \Gamma _k ^{\ast} \Big(B, \Gamma _{n-k} ^{\ast} (C,A)\Big) _{\widetilde{m} _1 \widetilde{m} _2} \Big(Y_{bca} \Big)^{\widetilde{m} _1 \widetilde{m} _2} _{m_1 m_2} + \nonumber \\
+ \Gamma _k ^{\ast}
\Big(C, \Gamma _{n-k} ^{\ast} (A,B)\Big) _{\widehat{m} _1 \widehat{m}
  _2} \Big(Y_{cab} \Big)^{\widehat{m} _1 \widehat{m} _2}
_{\widetilde{m} _1 \widetilde{m} _2} \Big(Y_{bca} \Big)^{\widetilde{m}
  _1 \widetilde{m} _2} _{m_1 m_2} \Bigg\} = 0
\nonumber \\
\end{eqnarray}
We want to exclude from our considerations the ``trivial'' deformations, i.e., the simple
$\lambda$-dependent changes of the basis $Q^{\ast}: V \rightarrow V$, such that:
\begin{eqnarray}
 \Gamma ^{\ast}(A,B, \lambda)_m = Q^{\ast -1}  (\Gamma ^{\ast}(Q^{\ast}A,Q^{\ast}B)_m), \quad Q^{\ast}= 1 + \lambda q_1^{\ast} + \lambda ^2 q_2^{\ast} +...
\end{eqnarray}
Written in a series over $\lambda$ up to first order, the deformed bracket becomes:
\begin{eqnarray}
 Q^{\ast -1}  (\Gamma ^{\ast}(Q^{\ast}A,Q^{\ast}B) _m)& = & (1-\lambda q_1^{\ast}) \Gamma _0 ^{\ast}\Big( (1+\lambda q_1^{\ast})A,(1+\lambda q_1^{\ast})B\Big) _m+ 0(\lambda ^2) \nonumber \\
& = & \Gamma ^{\ast} _0 (A,B)_m + \lambda \Gamma ^{\ast} _1 (A,B)_m +0(\lambda^2) \\[2mm]
\Gamma ^{\ast} _1 (A,B)_m & = & \Gamma ^{\ast} _0 (A,q_1^{\ast}B)_m + \Gamma ^{\ast} _0 (q_1^{\ast}A,B)_m -q_1^{\ast}\Gamma ^{\ast} _0 (A,B)_m
\end{eqnarray}

So, we have to ``factorize'' the set of admissible deformations over
the set of trivial deformations. In the case of associative algebra
such a factorization gave the opportunity to relate the deformations
and the conditions for the $i^{\rm th}$-order perturbation to a
Hochschild cohomology complex. We will show that also in our case the
deformations are described in terms of a cohomology complex, namely
the reduced Lie algebra complex from the previous section. 

In the following we formulate in cohomological language some of the
formulas above: 
\begin{observation}[1]
The Jacobi identity for $\Gamma_0^*$ can be rewritten in the compact
form: 
\begin{equation}
 (b^2 \Gamma_0^{\ast}) (A,B,C) _{m_1 m_2} = 0\,.
\end{equation}
Here, and in the following, the differentials $b^n$ of the chain
complex (\ref{eq:dif1}) are defined with $\Gamma^\ast_0$. 
\end{observation}

\begin{observation}[2]
When we insert the formal power series (\ref{eq:pert}) into the Jacobi
identity for the deformed bracket $\Gamma^\ast$, we get in first order
the following restriction on the first  order perturbation: 
\begin{equation}
 (b^2 \Gamma ^{\ast} _1 )(A,B,C) _{m_1 m_2} = 0\, ,
\end{equation}
i.e., $\Gamma ^{\ast} _1 \in Z^2(V)$. The first trivial perturbation is: 
\begin{equation}
 \Gamma ^{\ast} _1 (A,B) _m = (b^1 q ^{\ast})(A,B)_m
\end{equation}
which means $ \Gamma ^{\ast} _1 \in B^2(V)$. Then it follows, that the
non-trivial first order perturbations correspond to non-trivial
classes $[\Gamma_1^{\ast}] \in RLH^2(V)$.
\end{observation}

\begin{observation}[3]
The terms in the Jacobi identity (\ref{eq:JIdefj}) involving
$\Gamma^\ast_n$, i.e., those with $k=0,n$, precisely equal $b^2
\Gamma ^ {\ast}_n (A,B,C) _{m_1 m_2}$. One can therefore write 
(\ref{eq:JIdefj}) as an equation for $\Gamma^\ast_n$:
\begin{eqnarray}
\label{eq:JId1}
 b^2 \Gamma ^ {\ast}_n (A,B,C) _{m_1 m_2}  =
- \, \sum _{k=1} ^{n-1} \, \Bigg\{\Gamma _k ^{\ast} \Big(A,
\Gamma_{n-k}^{\ast} (B,C)\Big) _{m_1 m_2} + \hskip50mm \\ + \Gamma_k^{\ast} \Big(B,\Gamma_{n-k}^{\ast} (C,A)\Big) _{\widetilde{m} _1 \widetilde{m} _2} \Big(Y_{bca} \Big)^{\widetilde{m} _1 \widetilde{m} _2} _{m_1 m_2} 
+ \Gamma_k^{\ast}
\Big(C,\Gamma_{n-k}^{\ast} (A,B)\Big)_{\widehat{m}_1 \widehat{m}_2} \Big(Y_{cab} \Big)^{\widehat{m}_1 \widehat{m}_2}
_{\widetilde{m}_1 \widetilde{m}_2} \Big(Y_{bca} \Big)^{\widetilde{m}
  _1 \widetilde{m} _2} _{m_1 m_2} \Bigg\}\,. \nonumber
\end{eqnarray}

\end{observation}

The interesting question is whether every first order perturbation in
$RLH^2(V)$ is integrable, i.e., whether every $\widetilde{\Gamma}_1 ^{\ast}\in
Z^2(V)$ serves as the first order perturbation for some one-parameter
family of deformations of $\Gamma^\ast_0$. One has to decide whether
the equations (\ref{eq:JId1}) can be solved recursively with a given
$\widetilde{\Gamma}_1 ^{\ast}$.   

Thus suppose that for some $n\geq 2$, candidates
$\widetilde\Gamma^\ast_j\in C^2(V)$ ($2\leq j<n$) for the coefficients
of a perturbative expansion have been found 
solving (\ref{eq:JId1}) for $n'<n$. Then $\widetilde\Gamma^\ast_n$ must solve 
\begin{equation} \label{eq:ntho}
b^2 \widetilde\Gamma ^ {\ast}_n (A,B,C) _{m_1 m_2}  = 
G^n[\widetilde{\Gamma} ^{\ast}_1,...,\widetilde{\Gamma}
^{\ast}_{n-1}](A,B,C)_{m_1 m_2} 
\end{equation}
where $G^n[\widetilde{\Gamma} ^{\ast}_1,...,\widetilde{\Gamma}
^{\ast}_{n-1}](A,B,C)_{m_1 m_2}$ is the r.h.s.\ of (\ref{eq:JId1})
evaluated on the lower order perturbations $\widetilde\Gamma^\ast_j$: 
\begin{eqnarray}\label{Gn}
G^n[\widetilde{\Gamma} ^{\ast}_1,...,\widetilde{\Gamma}
^{\ast}_{n-1}](A,B,C)_{m_1 m_2} 
:=
- \, \sum _{k=1} ^{n-1} \, \Bigg\{ \widetilde\Gamma _k ^{\ast} \Big(A, \widetilde\Gamma
_{n-k} ^{\ast} (B,C)\Big) _{m_1 m_2} + \hskip30mm \\ +  \widetilde\Gamma _k ^{\ast} \Big(B, \widetilde\Gamma _{n-k} ^{\ast} (C,A)\Big) _{\widetilde{m} _1 \widetilde{m} _2} \Big(Y_{bca} \Big)^{\widetilde{m} _1 \widetilde{m} _2} _{m_1 m_2} 
+ \widetilde\Gamma _k ^{\ast}
\Big(C, \widetilde\Gamma _{n-k} ^{\ast} (A,B)\Big) _{\widehat{m} _1 \widehat{m}
  _2} \Big(Y_{cab} \Big)^{\widehat{m} _1 \widehat{m} _2}
_{\widetilde{m} _1 \widetilde{m} _2} \Big(Y_{bca} \Big)^{\widetilde{m}
  _1 \widetilde{m} _2} _{m_1 m_2} \Bigg\} \nonumber
\end{eqnarray}
This clearly exhibits the role of $G^n$ as ``obstruction operators'':
namely the equation (\ref{eq:ntho}) for $\widetilde\Gamma^{\ast}_n$ is 
consistent only if
\begin{eqnarray}\label{eq:obst}
G^n[\widetilde\Gamma^\ast_1,\dots,\widetilde\Gamma_{n-1}^\ast]\in B^3(V),
\end{eqnarray}
in which case $\widetilde\Gamma^{\ast}_n$ is determined up to an
element of $Z^2(V)$ (that one may absorb into $\widetilde\Gamma^{\ast}_1$ by a redefinition of the perturbation parameter $\lambda$).

In analogy to deformation theory of associative algebras, we expect
that for $\widetilde\Gamma^\ast_1,\dots,\widetilde\Gamma^\ast_{n-1}\in
C^2(V)$ solving (\ref{eq:JId1}) for all $n'<n$, in particular
$\widetilde\Gamma^\ast_1\in Z^2(V)$, one always has
\begin{eqnarray} \label{eq:bG0}
b^3G^n[\widetilde\Gamma^\ast_1,\dots,\widetilde\Gamma_{n-1}^\ast]=0,\quad\hbox{i.e..}\quad
G^n[\widetilde\Gamma^\ast_1,\dots,\widetilde\Gamma_{n-1}^\ast] \in Z^3(V).
\end{eqnarray}
The cohomology class of
$G^n[\widetilde\Gamma^\ast_1,\dots,\widetilde\Gamma_{n-1}^\ast]$ in 
$RLH^3(V)$ is referred to as an obstruction. For (\ref{eq:ntho}) to
have a solution, the obstruction must be zero. If it happens that
$RLH^3(V)$ is trivial, no obstructions can occur in any order, and all
$\widetilde\Gamma^{\ast}_n$ can be found recursively from a given
$\widetilde\Gamma^{\ast}_1\in Z^2(V)$, i.e., every first-order
perturbation is integrable. If $RLH^3(V)$ is nontrivial, then some
first-order perturbations may still be integrable, but (\ref{eq:obst})
might impose further restrictions. 

\medskip

To address and decide the possibility of (formal) continuous
deformations of a given reduced Lie algebra, one therefore has to
compute its second and third cohomologies (provided
(\ref{eq:bG0}) can be established). This is outside the scope of this
article.

\section{Outlook}

We showed that the commutation relations among quasiprimary fields in
conformal chiral field theories ($W$-algebras) are fixed
up to structure constants that are related to the $3$-point
amplitudes. We then explicitly exhibited an
infinite number of constraints on the structure constants, which
warrant anti--symmetry and Jacobi identity for the commutator of field
operators, and positivity of the Hilbert space inner product. Their
solutions can therefore be used as a new axiomatization of chiral CFT. 
%For example, we expect that they relate
%proportionally certain couples of structure constants or that they
%imply that a certain subset of the structure constants vanish. 
It is not a surprise that in the easiest case, the solution
of the constraints on the structure constants for
fields of dimension $1$ reproduces the well-known Kac-Moody algebras,
including the necessary compactness of the underlying Lie algebra. It
remains to analyze these constraints more carefully in the general case.  

In more abstract language, the structure constants define a bracket on 
a reduced field space. We proceeded to explore the rigidity of this
bracket under formal deformations, in other words to check whether
there exist models in the neighborhood of a given model. Following the
general strategy, we constructed a cohomology complex related to the
deformation problem, and showed that the cohomology groups
$RLH^2(V)$ and $RLH^3(V)$ determine the existence and integrability of 
deformations (with the proviso that (\ref{eq:bG0}) was not yet proven). We have, though, not been able to actually compute the
cohomology groups associated to this complex and 
this has to be done before a more complete deformation theory can be
developped. Another option would be to try to construct a differential 
graded Lie algebra out of the cohomology complex, whose deformation
theory would be tightly related to the deformation theory of the
reduced bracket. For this purpose, one has to construct a bracket in
this complex, such that it is skew symmetric with respect to the
grading by dimension of the cochain spaces and satisfies a graded
Jacobi identity.  

As an example what the deformation theory would produce when fully
worked out, one may think of the theory generated by the stress-energy
tensor, which has the central charge $c$ as a free parameter. The
number of composite fields of a given dimension is determined by a
well-known character formula, so this would fix the multiplicities and
hence the reduced space. But for $c<1$ the presence of zero-norm
vectors in the Verma module reduces the multiplicities. We therefore
expect that for $c>1$, the second cohomology $RLH^2(V)$ is nontrivial,
admitting an infinitesimal change of $c$, and $RLH^3(V)$ could be
trivial as there is no obstruction against finite variations of
$c$; on the other hand, for $c<1$ and with the reduced multiplicities,
the second cohomology is expected to be trivial. Of course, presently
we cannot establish these claims ``from scratch''.  

In previous approaches \citep{zam86,bouw88,blum91}, $W$-algebras were
analysed in terms of finite sets of fields, which generate the
infinite space of quasiprimary fields under the OPE. Indeed, the
consistency of the commutation relations can be studied at the level
of the generating fields; but -- as the example of the stress-energy
tensor shows -- to address the issue of positivity, one has to include
all their composite fields. In our approach, no distinguished role is
assigned to the generating fields, except that they could possibly be a
practical tool for solving the constraints in an inductive way.  

The cohomological nature of the deformation theory was recognized 
previously and turned into a constructive tool, e.g., for
perturbations of free fields \citep{hollands08}; for the
classification of $W$-algebras it was not exploited yet.  

%\section*{Acknowledgments}

\begin{appendix}
\section{Appendix: The matrix $\big(Y_{abc} \big) ^{m _1 m _2} _{\widetilde{m}_1 \widetilde{m}_2}$}
\label{sec:app}

In this subsection we will explain how we determined the matrix
$\big(Y_{abc} \big)_{\widetilde{m}_1 \widetilde{m}_2}^{m_1 m_2}$
which transforms $\big(T_{cab}\big)^{\widetilde{m}_1
  \widetilde{m}_2}(h,f,g)$ into $\big(T_{abc}\big)^{m_1 m_2}(f,g,h)$,
and which is the essential ingredient of the reduced Jacobi
identity (\ref{eq:JI1}).

An elegant method would have been to exploit the associativity of a
nontrivial one-parameter family of products on $\bigoplus_k M_{2k}$
defined in terms of Rankin--Cohen brackets \citep{cmz96}, generalizing
an unpublished observation by Eholzer. Varying the parameter, one
obtains linear relations between 
$\lambda_{dc}^e\circ(\lambda_{ab}^d\times 1_c)$ and $\lambda_{ad'}^e\circ
(1_a\times\lambda_{bc}^{d'})$ for every fixed $a,b,c,e$, from which one would 
read off the matrix $X_{abc}$ of Sect.\ \ref{sec:XY}
that describes the re-bracketing, and then by (\ref{X=YI}) the matrix
$Y_{abc}$. Unfortunately, due to a symmetry with respect to the
parameter, varying the parameter gives only one half of the necessary
relations. This is a bit of a surprise since one would have naively
expected that an associative product rather encodes twice as much
information that a (generalized) commutator. 

Instead, we have to adopt a much more down-to-earth linear algebra
approach. By applying the intertwiners to test functions and comparing
the resulting coefficients of products of derivatives, allowed us to
derive a recursion formula for the entries of $\big(Y_{abc}
\big)_{\widetilde{m}_1 \widetilde{m}_2}^{m_1 m_2}$, which we were able
to solve afterwards.  

The explicit formulae thus obtained below are meromorphic functions
which may have poles at real positive values of the dimensions $a,b,c$. 
In other words, the intertwiner bases may become degenerate at these
points. These singularities can be regularized, e.g., by letting the
scaling dimensions have small positive imaginary parts, while keeping the
summation indices $p,q$ in (\ref{eq:lamfg}) and $m$ in (\ref{eq:basT})ff
integer. While the representation theory of $\hbox{SL}(2,\mathbb{R})$
is perfectly meaningful for complex $a,b,c$, the physical dimensions
are of course positive integers. For the removal of the regularization
in QFT, see Sect.\ \ref{sec:RJI}. 

Using (\ref{eq:lambda}) we write the explicit expression for the
composite intertwiners:

\begin{eqnarray} \label{eq:T1}
\Big(T _{abc}\Big)^{m _1 m _2} (f,g,h) & = & \lambda ^{e_1} _{a e_2} 
\circ \big( 1_a \times \lambda ^{e_2} _{bc} \big) (f,g,h) \nonumber \\
& = & \!\!\!\sum _{\substack{p+q = m_1\\s+t=m_2}} (-1)^{q+t}  
\frac{(2b + 2c -m_1 -2m_2 -3)_p}{p!}\frac{(2a -m_1 -1)_q}{q!}
\times \nonumber \\
\label{eq:Vt}&& \times 
\frac{(2b -m_2 -1)_t}{s!} \frac{(2b -m_2 -1)_t}{t!} 
\partial ^p f \, \partial ^q (\partial ^s
g \, \partial ^t h),
\end{eqnarray}
which can be expanded as
\begin{eqnarray}
\Big(T _{abc}\Big)^{m _1 m _2} (f,g,h) & = & \sum _ {r_1 + r_2 + r_3 =
  m_1 + m_2}\Big(T _{abc}\Big)^{m _1 m _2} _{r_1 r_2 r_3} \partial
^{r_1}f \partial ^{r_2}g \partial ^{r_3}h 
\end{eqnarray}
with $\big(T _{abc}\big)^{m _1 m _2} _{r_1 r_2 r_3} $ numerical coefficients, 
and similar for $ \big(T _{cab }\big)^{\widetilde{m} _1\widetilde{m}_2}  (h,f,g)$.
We therefore have to solve, for any fixed triple $(r_1,r_2,r_3)$, the equation
\begin{equation}
\Big(Y_{abc} \Big) _{\widetilde{m} _1 \widetilde{m} _2} ^{m_1 m_2}\Big(T_{cab}\Big)^{\widetilde{m} _1 \widetilde{m} _2} _{r_1 r_2 r_3} = \Big(T_{abc}\Big)^{m _1 m _2} _{r_1 r_2 r_3}\,.
\end{equation}

The following two observations now simplify the problem.
\begin{enumerate}
\item Because $r_1 + r_2 + r_3 =   m_1 + m_2$, we must have $m_1 + m_2
  = \widetilde{m} _1 + \widetilde{m} _2$. The matrix 
$\big(Y_{abc} \big) _{\widetilde{m}_1 \widetilde{m}_2}^{m_1 m_2}$
therefore has a block form, reflecting the fact, that it is not 
  possible to decompose $\big(T_{abc}\big)^{m _1 m _2}$ in the basis
  of $\big(T_{cab}\big)^{\widetilde{m} _1 \widetilde{m} _2}$ if they
  map to representations with different scaling dimensions. Then we
  can relax two of the indices of $\big(Y_{abc} \big) _{\widetilde{m}
    _1 \widetilde{m} _2} ^{m_1 m_2}$:  
\begin{equation}
\Big(Y_{abc} \Big) _{\widetilde{m} _1 \widetilde{m} _2} ^{m_1 m_2} = \delta _{m_1 + m_2,\widetilde{m} _1 + \widetilde{m} _2} \Big(Y_{abc} (n)\Big) _{\widetilde{m} _2} ^{m_2}, \quad n:=m_1+m_2
\end{equation}
(We could as well relax the indices $m_2$ and $\widetilde{m}_2$ instead of $m_1$ and $\widetilde{m}_1$, it is just a matter of choice.) Then, we have to solve
for any fixed triple $(r_1,r_2,r_3)$ 
\begin{equation}
\label{eq:y(n)}
\Big(Y_{abc} (n)\Big) _{\widetilde{m} _2} ^{m_2}\Big(T_{cab}\Big)^{n-\widetilde{m}_2, \widetilde{m}_2}_{r_1r_2r_3} = \Big(T_{abc}\Big)^{n-m_2, m_2}_{r_1r_2r_3}\,.
\end{equation}
\item{(\ref{eq:y(n)}) taken for $n+1$ triples $(r_1,r_2,r_3)$ and
    $m_2, \, \widetilde{m} _2 \in[0,n]$ gives a system of $(n+1)
    \times (n+1)$ equations for $(n+1)\times (n+1)$ 
unknown quantities, and if these equations are linearly independent it
is enough to fix all the entries of $\big(Y_{abc} (n)\big)_{\widetilde{m} _2} ^{m_2}$.
A most convenient choice are the triples
$(k,0,n-k)$ with $k \in [0,n]$, because the coefficients
$\big(T_{cab}\big)^{n-\widetilde{m}_2, \widetilde{m}_2}_{k,0,n-k}$,
read off from (\ref{eq:T1}), 
are zero if $\widetilde{m}_2> k$. This allows to establish a recursion, such that the component 
$\Big(Y_{abc} (n)\Big) _{\widetilde{m} _2} ^{m_2}$ is obtained recursively from the components $\Big(Y_{abc} (n)\Big) _{\widehat{m} _2} ^{m_2}$ with $\widehat{m}_2 < \widetilde{m}_2$.}
\end{enumerate}

\begin{proposition}
 The entries of $\big(Y_{abc} (n)\big) _{\widetilde{m} _2} ^{m_2}$ satisfy the recursion formula:
\begin{eqnarray}
\label{eq:Yrec}
 \Big(Y_{abc} (n)\Big) _{\widetilde{m} _2} ^{m_2} 
& = & \frac{\Big(T_{abc}\Big)^{n-m_2, m_2}_{\widetilde{m} _2,0,n-\widetilde{m} _2} - \sum _{j=0} ^{\widetilde{m} _2-1}
\Big(Y_{abc} (n)\Big) _{j} ^{m_2} \Big(T_{cab}\Big)^{n-j, j}_{\widetilde{m} _2,0,n-\widetilde{m} _2}}{\Big(T_{cab}\Big)
^{n-\widetilde{m} _2, \widetilde{m} _2}_{\widetilde{m} _2,0,n-\widetilde{m} _2}}
\end{eqnarray}
\end{proposition}

To solve this recursion, we ``insert repeatedly this formula into itself'' and obtain:

\footnotesize
\begin{eqnarray}\label{recurs}
 \Big(Y _{abc} (n)\Big) _{\widetilde{m} _2} ^{m_2} & = & \sum \limits
 _{s=0} ^{\widetilde m_2}\frac{ \Big(T_{abc}\Big)^{n-m_2, m_2}_{s,0,n-s} }{\Big(T_{cab}\Big)
^{n-\widetilde{m} _2, \widetilde{m} _2}_{\widetilde{m} _2,0,n-\widetilde{m} _2}}\Bigg[(-1)^{\widetilde{m} _2 -s}\frac{(n-s)!}{(n-\widetilde{m} _2 )!}
\frac{\left(2c-(n-s) -1 \right)_{\widetilde{m} _2 -s}}
{\left(2a+2b -2\widetilde{m} _2 -3 \right)_{2\widetilde{m} _2 -2s}} 
\times \nonumber \\
&&  \hskip20mm\times \sum \limits _{\substack{\{j_l\}_{s} ^{\widetilde{m} _2}}} (-1)^{l-1}  
\prod_{j_r \in \{j_l\}^{\widetilde{m}_2} _s} \frac{(2a+2b -2j_{r+1} -3)_{j_{r+1} - j_{r}}}{(j_{r+1} - j_{r})!}\Bigg] \nonumber \\
\end{eqnarray}
\normalsize
where $\{j_l\}_{s} ^{m}$ are the possible sets $\{j_1=s,\, j_k <j_{k+1},\,j_l=m\}$, including $\{s,\, m\}$.
\begin{claim}
Extrapolating from calculations for small $l$, we found an explicit
identity to perform the multiple sum:
 
\footnotesize
\begin{equation}
\sum\limits_{\{j_l\}^{m}_s} (-1)^{l-1} \prod_{j_r \in \{j_l\}^{m}_s} 
\frac{(2a+2b -2j_{r+1} -3)_{j_{r+1} - j_{r}}}{(j_{r+1} - j_{r})!} = (-1)^{m-s} \frac{(2a+2b -2m  -3)
(2a+2b -m -s-2)_{m -s-1}}{(m -s)!}
\end{equation}
\normalsize
\end{claim}

Using this identity (with $m=\widetilde{m}_2$) we can reduce
(\ref{recurs}) to a single sum. We finally obtain
\begin{proposition}
The matrix $\big(Y _{abc} (n)\big) _{\widetilde{m} _2} ^{m_2}$ is
given by the following expression:

\footnotesize
\begin{eqnarray}
  \Big(Y _{abc} (n)\Big) _{\widetilde{m} _2} ^{m_2} & = & (-1)^{n-\widetilde{m}_2} \frac{{n \choose m_2}}{{n \choose \widetilde{m} _2}} \frac{(2-2b)_{m_2}}{(2-2b)_{\widetilde{m}_2}} \frac{1}{(2\widetilde{m}_2 - 2a-2b+4)_{n - \widetilde{m} _2}}\times  \\
&& \times \sum \limits _{s=0} ^{\widetilde{m}_2} 
\frac{{n-m_2 \choose s}(n+m_2-s- 2b-2c+4)_s(s-2a+2)_{n-m_2-s}{n-s \choose \widetilde{m} _2 -s}\left(n-\widetilde{m}_2 -2c+2 \right)_{\widetilde{m} _2 -s}}
{\left(2a+2b -2\widetilde{m} _2 - 2 \right)_{\widetilde{m} _2 -s}} \nonumber
\end{eqnarray}
\end{proposition}
This expression presumably cannot be further simplified, since the sum
does not factorize in general as a rational function of the dimensions.

We observed the following interesting property of the matrix $\big(Y _{abc} (n)\big)_{\widetilde{m} _2} ^{m_2}$:

\begin{proposition}
The entries from an arbitrary column of the matrix $\big(Y _{abc} (n)\big) _{\widetilde{m} _2} ^{m_2}$ sum to $(-1)^{n+\widetilde{m}_2}$, where $\widetilde{m}_2$ is the number of the column.
\end{proposition}

\begin{proof}
We will prove this statement by induction on the number of the column.

Let us first consider the column $\widetilde{m} _2 = 0$. The entries from this column are expressed as:
\begin{equation}
 \Big(Y _{abc} (n)\Big) _{0} ^{m_2} = (-1)^n \frac{(2-2a)_{n-m_2} }{(n-m_2)!} \frac{(2-2b)_{m_2}}{m _2!}\frac{n!}{[4-2a-2b]_n}
\end{equation}
Then, using the property $\frac{(a+b)_n}{n!}=\sum_{i=0} ^n \frac{(a)_i}{i!}\frac{(b)_{n-i}}{(n-i)!}$ we compute $\sum _{m_2 =0} ^n \Big(Y _{abc} (n)\Big) _{0} ^{m_2}=(-1)^n$, i.e., the statement of the proposition holds for $\widetilde{m} _2 = 0$.

Now let us assume that $\sum _{m_2 =0} ^n \big(Y_{abc}
(n)\big)_{j}^{m_2}=(-1)^{n+j}$ is true for every $j \leq k-1$. We will
prove that then $\sum_{m_2 =0} ^n \big(Y_{abc} (n)\big)_{k}
^{m_2}=(-1)^{n+k}$. We start from formula (\ref{eq:Yrec}) and obtain:

\footnotesize
\begin{eqnarray}
\label{eq:ind0}
\sum _{m_2 =0} ^n \Big(Y _{abc} (n)\Big) _{k} ^{m_2} & = &  \frac{1}{\Big(T_{cab}\Big)
^{n-k, k}_{k,0,n-k}}\left\{\sum\limits _{m_2=0}^{n-k}\Big(T_{abc}\Big)^{n-m_2, m_2}_{k,0,n-k}-\sum\limits_{j=0} ^{k} (-1)^{n+j} \Big(T_{cab}\Big)^{n-j, j}_{k ,0,n-k}\right\} + (-1)^{n+k} \nonumber \\
\end{eqnarray}
\normalsize
Hence, we have to prove that the expression in the brackets vanishes. Let us write this expression explicitly:
\footnotesize
\begin{eqnarray}
\label{eq:Srec}
&&\sum\limits _{m=0}^{n-k}\Big(T_{abc}\Big)^{n-m, m}_{k,0,n-k}-\sum\limits_{j=0} ^{k} (-1)^{n+j} \Big(T_{cab}\Big)^{n-j, j}_{k ,0,n-k} = \nonumber \\
&&=(-1)^{n+k}\bigg\{ \sum\limits_{m = 0} ^{n-k} \frac{(2b+2c -(n+m)-3)_k}{k!} \frac{(2a - (n-m)-1)_{n-m-k}}{(n-m-k)!} \frac{(2b-m-1)_m}{m!} - \nonumber \\
&& \phantom{=(-1)^{n+k}}- \sum\limits_{m = 0} ^{k} \frac{(2a+2b -(n+m)-3)_{n-k}}{(n-k)!} \frac{(2c - (n-m)-1)_{k-m}}{(k-m)!} \frac{(2b-m-1)_m}{m!}\bigg\}
\end{eqnarray}
\normalsize
With the identities $\frac{(A+B+1)_n}{n!}=\sum_{j=0}^n
\frac{(A+j+1)_{n-j}}{{(n-j)}!}\frac{(B-j+1)_{j}}{j!}$ and
$(a)_{m+n}=(a)_m(a+m)_n= (a)_n(a+n)_m$ one can prove that the first
sum is equal to the second sum in (\ref{eq:Srec}), hence the bracket in (\ref{eq:ind0}) vanishes: 
\begin{eqnarray}
\label{eq:ind}
\sum_{m_2 =0} ^n \Big(Y_{abc} (n)\Big)_{k} ^{m_2} & = & (-1)^{n+k}\,.
\end{eqnarray}
This proves the induction hypothesis and the proposition.
\end{proof}

\end{appendix}

\bibliographystyle{jfm} 
\bibliography{BibliographyP}
%\bibliography{Bibliography}

\end{document}